\documentclass[10pt,a4paper]{article}

\usepackage{chngcntr}
\usepackage{amsmath,amssymb,amsthm}
\usepackage{fullpage}
\usepackage{cite}
 \usepackage{color}
\usepackage{textcomp}
\usepackage{mathrsfs}
\numberwithin{equation}{section}

\DeclareMathOperator{\sgn}{sgn}
\DeclareMathOperator*{\argmax}{arg\,max}
\DeclareMathOperator*{\argmin}{arg\,min}
\theoremstyle{plain}
\newtheorem{theorem}{Theorem}[section]
\newtheorem{proposition}[theorem]{Proposition}

\newtheorem{lemma}[theorem]{Lemma}
\theoremstyle{definition}
\newtheorem{definition}[theorem]{Definition}
\newtheorem{assumption}[theorem]{Assumption}

\theoremstyle{remark}
\newtheorem{remark}[theorem]{Remark}
\newtheorem{example}[theorem]{Example}
\usepackage{parskip}
\makeatletter
\def\thm@space@setup{%
	\thm@preskip=\parskip \thm@postskip=5pt
}
\makeatother

\begin{document}
	\title{Portfolio liquidation under factor uncertainty\thanks{Financial support by {\sl d-fine GmbH} and the CRC 190 \textsl{Rationality and competition: the economic performance of individuals and firms} is gratefully acknowledged. }}

\author{Ulrich Horst\footnote{Department of Mathematics, and School of Business and Economics, Humboldt-Universit\"at zu Berlin
         Unter den Linden 6, 10099 Berlin, Germany; email: horst@math.hu-berlin.de} ~ Xiaonyu Xia \footnote{Department of Mathematics, Humboldt-Universit\"at zu Berlin
         Unter den Linden 6, 10099 Berlin, Germany; email: xiaxiaon@math.hu-berlin.de}~ and Chao Zhou\footnote{Department of Mathematics, National University of Singapore,  
 Lower Kent Ridge Road 10, 119076 Singapore, Singapore; email: matzc@nus.edu.sg}}

\maketitle

\begin{abstract}
We study an optimal liquidation problem under the ambiguity with respect to price impact parameters. Our main results show that the value function and the optimal trading strategy can be characterized by the solution to a semi-linear PDE with superlinear gradient, monotone generator and singular terminal value. We also establish an asymptotic analysis of the robust model for small amount of uncertainty and analyse the effect of robustness on optimal trading strategies and liquidation costs. In particular, in our model ambiguity aversion is observationally equivalent to increased risk aversion. This suggests that ambiguity aversion increases liquidation rates.   
\end{abstract}

{\bf AMS Subject Classification:} 93E20, 91B70, 60H30.

{\bf Keywords:}{~stochastic control, uncertainty, portfolio liquidation, singular terminal value, superlinear growth gradient}
\section{Introduction}

Starting with the work of Almgren and Chriss \cite{Almgren2001} optimal portfolio liquidation strategies under various market regimes and price impact functions have been analyzed by many authors. Single player models have been analyzed by \cite{Ankirchner2014, Becherer2017, Gatheral2011, Graewe2018, Graewe2015, Kratz2014, Popier2006}  among many others; multi-player models were analyzed in, e.g.~\cite{Bayraktar2017, FGHP-2018, HJN-2015}. From a mathematical perspective, the main characteristic of optimal liquidation models is the singular terminal condition of the value function that is induced by the liquidation constraint. The singularity becomes a major challenge when determining the value function and applying verification arguments. 

In this paper we study a class Markovian single-player portfolio liquidation problems where the investor is uncertain about the factor dynamics driving trading costs. The liquidation problem leads to a stochastic control problem of the form 
\begin{equation}\label{P1}
\inf_{\xi}\sup_{Q\in \mathcal Q}\Big(\mathbb E_{Q}\left[\int_0^T\eta(Y_s) |\xi_s|^p+\lambda(Y_s)|X_s|^p\,ds\right]-\Upsilon(Q)\Big)
\end{equation}
subject to the state dynamics
\begin{equation}\label{P2}
\begin{split}
	dY^{}_t &= b(Y^{}_t) dt + \sigma(Y^{}_t) dW_t, \quad Y^{}_0=y \\
	dX_t &=- \xi_t\,dt, \quad X_0=x 
\end{split}
\end{equation} 
and the terminal state constraint 
\begin{equation}\label{P3}
	X_{T}=0,
\end{equation}
where $\xi$ denotes the trading rate, $X$ denotes the portfolio process, $Y$ denotes a factor process that drives trading costs and $\mathcal{Q}$ is a set of probability measures that are absolutely continuous with respect to a benchmark measure $\mathbb{P}$. The functions $\eta$ and $\lambda$ specify the instantaneous market impact from trading and the market risk of a portfolio holding, respectively. Instead of restricting the set of probability measures ex ante, we add a penalty term $\Upsilon(Q)$ to the objective function. This approach was first introduced by Hansen and Sargent\cite{Hansen2001} and has since become a popular approach in both the economics and financial mathematics literature when analyzing optimal decision problems under model uncertainty.  

The benchmark case where $\mathcal Q$ contains a single element has been analyzed in \cite{Graewe2018, Horst2018}. In this case, the value function can be described in terms of the unique nonnegative viscosity solution of polynomial growth of a semi-linear PDE with singular terminal value. The proof is based on an asymptotic expansion of the solution around the terminal time that shows that the value function converges to the instantaneous impact factor at the terminal time when properly rescaled.   

If $\mathcal{Q}$ contains more than one element, then the investor is uncertain about the dynamics of the factor process. For instance, the process $\eta(Y_t)$ may be viewed as describing the inverse market depth, whose dynamics the investor may not be able to specify correctly. The market risk factor $\lambda(Y_t)$, on the other hand, can be linked to the volatility of the reference price process. If the price dynamics follows a stochastic volatility model, then factor uncertainty amounts to uncertainty about the volatility of the reference price. 

{
Under factor uncertainty additional regularity assumptions on the penalty function $\Upsilon(Q)$ are required to guarantee that the optimization problem is tractable analytically. In order to guarantee analytical tractability we follow an approach that had first been introduced by Maenhout \cite{Maenhout2004} when analyzing a class of portfolio allocation models for Merton-type investors under model uncertainty.\footnote{The approach has been adapted by many authors, including \cite{Branger2013, Flor2013, Munk2013, Escobar2015, Zeng2018}, partly due to its analytical tractability but also due to the ``embedded'' equivalence between ambiguity and risk aversion.} Specifically, we consider penalty functions with state-dependent ambiguity aversion parameters that satisfy a scaling property corresponding to homothetic preferences. The assumption of homothetic preferences does not only facilitate the mathematical analysis but it also has a clear economic implication. Our model with ambiguity aversion is observationally equivalent to a model without ambiguity aversion but increased risk aversion. An approach that is similar in spirit to the ones in \cite{Maenhout2004} and in this paper has been followed by Bj\"{o}rk  et al.\cite{Bjoerk2012}. They studied an equilibrium model with mean-variance preferences and a (state-dependent) dynamic risk aversion parameter that is inversely proportional to wealth. For their choice of risk aversion the equilibrium monetary amount invested in the risky asset is proportional to current wealth.}


%

Under our scaling property on the penalty function, we prove that the value function to our control problem can be characterized by the solution to a semi-linear PDE with superlinear gradient, monotone generator and singular terminal value. Our first main contribution is to prove that this PDE admits a unique nonnegative viscosity solution of polynomial growth under standard assumptions on the factor process and the cost coefficients. The dependence of the generator on the gradient requires additional regularity properties of the viscosity solution in order to carry out the verification argument. Under an additional assumption on the penalty function and an additional boundedness condition on the market impact term we prove that the viscosity solution is indeed of class $C^{0,1}$. The proof is based on an asymptotic expansion of the solution around the terminal time as in \cite{Graewe2018, Horst2018} with the added difficulty that now not only the value functions but also its derivative needs to converge to the market impact term, respectively its derivative when properly rescaled.   
 
 The additional regularity of the solution does not only allow us to obtain the optimal trading strategy but also the least favourable martingale measure in feedback form.  For small amounts of uncertainty it also allows us to provide a first order approximation of the value function in terms of the solution to the benchmark model without uncertainty. Finally, we prove that our model with factor uncertainty is observationally equivalent to a model without factor uncertainty but increased market risk. This suggests that factor uncertainty increases the rate of liquidation.

To the best of our knowledge, only few papers have studied the optimal liquidation problem under model uncertainty. Nystr\"{o}m et al. \cite{Nystrom2014} and Cartea et al. \cite{Cartea2017, Cartea2017a} considered problems of optimal liquidation with \textsl{limit orders} for a CARA, restectively a risk-neutral investor. In \cite{Nystrom2014} it is assumed that the investor is uncertain about both the drift and the volatility of the underlying reference price process. They show that uncertainty may increase the bid-ask spread and hence reduce liquidity. In \cite{Cartea2017, Cartea2017a} the investor is uncertain about the arrival rate of market orders, the fill probability of limit orders and the dynamics of the asset price. They show that ambiguity aversion with respect to each model factor has a similar effect on the optimal strategy, but the magnitude of the effect depends on time and inventory position in different ways depending on the source of uncertainty. In both papers strict liquidation is not required; instead open positions at the terminal time are penalized. This avoids the mathematical challenges resulting from the singular terminal value.

Lorenz and Shied \cite{Lorenz2013} studied the drift dependence of optimal trade execution strategies under transient price impact with exponential resilience and strict liquidation constraint. They find an explicit solution to the problem of minimizing the expected liquidation costs when the unaffected price process is a square-integrable semimartingale. Later, Schied \cite{Schied2013b} analysed the impact on optimal trading strategies with respect to misspecification of the law of the unaffected price process in a model which only allows instantaneous price impact. Both papers studied the dependence of optimal liquidation strategies on model dynamics but did not consider the resulting robust control problem. Bismuth et al. \cite{Bismuth2019} considered a portfolio liquidation model for a CARA investor that is uncertain about the drift of the reference price process but did not require a strict liquidation constraint. They do not consider a robust optimization problem either but dealt with the uncertainty by a general Bayesian prior for the drift, which allows them to solve the problem by dynamic programming techniques.  All three papers focussed on misspecification of the reference price process and assumed that the market impact parameters are known. Our model is different; we analyze the effect of uncertainty about the model parameters, e.g.~the market depth that we consider the most important impact factor.      

In a recent paper, Popier and Zhou \cite{Popier2019} analysed the optimal liquidation problem under drift and volatility uncertainty in a non-Markovian setting and characterized the value function by the solution of a second-order BSDE with monotone generator and singular terminal condition. In contrast to \cite{Popier2019}, we focus on the drift uncertainty about the factor model and add a penalty function in the spirit of convex risk measure theory. We also obtain much stronger regularity properties of the value function which allows us to study the effect of uncertainty on optimal trading strategies and costs in greater detail.

The remainder of this paper is organized as follows. In Section \ref{formulation}, we describe the modelling set-up, introduce the stochastic control problem and state our main results. The existence of viscosity solution to the HJBI equation is established in Section \ref{viscosity solution}; the regularity of the viscosity solution is proved in Section \ref{regularity}. The verification argument is carried out in Section \ref{verification}. Finally,  Section \ref{asmyptotic analysis} is devoted to an asymptotic analysis of the value function for small amounts of uncertainty. 

	\textit{Notation and notational conventions.}
We put $$\langle y\rangle:=(1+|y|^2)^{1/2}.$$ Let $I$ be a compact subset of $\mathbb R$.  We denote by $C_b(\mathbb R^d), C_b(I\times\mathbb R^d)$ the spaces of bounded continuous functions on $\mathbb R^d$, respectively, $I\times\mathbb R^d$.  For a given $n\geq0,$ we define $C_n(\mathbb R^d)$ (resp. $C_n(I\times\mathbb R^d)$) to be the set of functions $\phi\in C(\mathbb R^d)$ (resp. $C(I\times\mathbb R^d)$) such that $$\psi:=\frac{\phi(y)}{1+|y|^n} \in C_b(\mathbb R^d) (\text{resp. }\psi:=\frac{\phi(t,y)}{1+|y|^n} \in C_b(I\times\mathbb R^d)).$$ 
A function $\phi$ belongs to $USC_n(I\times\mathbb R^d)$ (or $LSC_n(I\times\mathbb R^d))$ if it has at most polynomial growth of order $n$ in the second variable uniformly with respect to $t\in I$ and is upper (lower) semi-continuous on $I\times\mathbb R^d$. Denote by $C^{0,1}(I\times\mathbb R^d)$ the set of all functions $\phi:I\times\mathbb R^d\rightarrow\mathbb R$ which are continuous and continuously differentiable with respect to the second variable on $I\times\mathbb R^d$.  

The spaces $L^q_\mathcal F(0,T;\mathbb R^d), H^q_\mathcal F(0,T;\mathbb R^d)$ denote the sets of all the adapted processes $(Z_t)_{t\in[0,T]}$ satisfying that $\mathbb E[\int_0^T|Z_t|^q\,dt)]<\infty$, $\mathbb E[(\int_0^T|Z_t|^2\,dt)^{q/2}]^{1/q}<\infty$, respectively; the subet of processes with continuous paths satisfying $\mathbb E[\sup_{t\in[0,T]}|Z_t|^{q/2}]^{1/q}<\infty$ is denoted by $\mathcal S^{q}_{\mathcal F}(\Omega;C([0,T];\mathbb R^d))$.
Whenever the notation $T^-$ appears in the definition of a function space we mean the set of all functions whose restrictions satisfy the respective property when $T^-$ is replaced by any $s<T$, e.g., 
\[
	C_n([0,T^-]\times\mathbb R^d)=\{u:[0,T)\times\mathbb R^d\rightarrow \mathbb R: u_{|[0,s]\times\mathbb R^d}\in C_n([0,s]\times\mathbb R^d) \text{ for all } s\in[0,T)\}.
\]
Throughout, all equations and inequalities are to be understood in the a.s.\ sense. We adopt the convention that $C$ is a constant that may vary from line to line  and the operator $D$ denotes the gradient with respect to the space variable. 
\section{Problem formulation and main results}\label{formulation}
Let $T \in (0,\infty)$ and let $(\Omega,\mathcal F,(\mathcal F_t)_{t\in[0,T]},\mathbb P)$ be a filtered probability space that satisfies the usual conditions and carries an $n$-dimensional standard Brownian motion $W$ and an independent one-diemensional standard Brownian motion $B$. 

In this paper we consider the problem of a large investor that needs to liquidate a given portfolio $x\in\mathbb R$ within the time horizon $[0,T]$. Let $t\in[0,T)$ be a given point in time and $x\in \mathbb R$ be the portfolio position of the trader at time $t$. We denote by $\xi_s \in \mathbb R$ the rate at which the agent trades at time $s\in[t,T)$. Given a trading strategy $\xi$, the portfolio position at time $s\in [t,T)$ is given by
	\begin{equation*} 
	X_s= x - \int_t^s \xi_r\,dr,  \quad s \in [t,T]
	\end{equation*}
	and the liquidation constraint is
	\begin{equation} \label{liquidation constraint}
	X_T= 0.
	\end{equation}
In what follows we assume that all trading costs are driven by a factor process given by the $d$-dimensional It\^{o} diffusion
 \begin{equation*}
 \left\{
 \begin{aligned}
	dY^{t,y}_s &= b(Y^{t,y}_s) ds + \sigma(Y^{t,y}_s) dW_s, \quad s\in[t,T],\\
	Y^{t,y}_t&=y.
\end{aligned}\right.
\end{equation*} 

Our goal is to analyze the impact of uncertainty about the factor dynamics on optimal liquidation strategies and trading costs. 
	
\subsection{The benchmark model}

In this section we briefly recall the liquidation model without factor uncertainty analyzed by Graewe et al. \cite{Graewe2018} against which our results shalll be benchmarked. Following \cite{Graewe2018}, we assume that the investor's transaction price $P_s\in\mathbb R$ at time $s \in [t,T]$ can additively decomposed into a fundamental asset price $\tilde P_s$ and an instantaneous price impact term $f(\xi_s)$ as
	\begin{equation*}
	P_s=\tilde{P}_s - f(\xi_s)
	\end{equation*}
where the fundamental asset price process $\tilde P$ is given by a one-dimensional square-integrable Brownian martingale, which we assume to be of the form\footnote{See Example \ref{ex1} below for a stochastic volatility model with uncertainty about the driver of the volatility process. } 
	\begin{equation*}
d\tilde P_s=\tilde{\sigma}(Y^{t,y}_s) dB_s
	\end{equation*}
	for some function $\tilde{\sigma}$. The investor aims at minimizing the difference between the book value of the portfolio and the expected proceeds from trading  plus risk cost. 
	We assume that the instantaneous impact factor is given by $f(\xi_s) = \eta(Y^{t,y}_s )|\xi_s|^{p-1}\sgn(\xi_s)$ for some $p>1$ and some bounded function $\eta$ that describes the inverse market depth and that the risk is measured by the integral of the $p$-th power of the value at risk of an open position  over the trading period. The resulting cost functional is then given by
	\begin{equation}\label{cost}
	\begin{aligned}
	J(t,y,x,\xi)&=\textrm{book value}-\textrm{expected proceeds from trading}+\textrm{risk costs}\\
&=\mathbb E_{\mathbb P}\Big[\int^T_t\eta(Y^{t,y}_s) |\xi_s|^pds+\int^T_t X_sd\tilde{P}_s+\int^T_t\lambda(Y^{t,y}_s)|X_s|^p\,ds\Big] \\
	&=\mathbb E_{\mathbb P} \Big[\int_t^T\big( \eta(Y^{t,y}_s) |\xi_s|^p+\lambda(Y^{t,y}_s)|X_s|^p\big)\,ds\Big],
	\end{aligned}
	\end{equation}
where the last equality follows from the facts that $X\in \mathcal S^{2}_{\mathcal F}(\Omega;C([t,T];\mathbb R^d))$ and that $\tilde P$ is a square-integrable martingale under $\mathbb P$.

For each initial state $(t,y,x)\in[0,T)\times\mathbb R^{d}\times\mathbb R$ the value function of the investor's control problem is defined by
\begin{equation} \label{value-function-0}
	V_0(t,y,x):=\inf_{\xi\in\mathcal A(t,x)}J(t,y,x,\xi)
\end{equation}
where the infimum is taken over the set $\mathcal A(t,x)$ of all admissible controls, that is, over all the controls $\xi$ that  belong to $ L_{\mathcal F}^{2p}(t,T;\mathbb R)$ and that satisfy the liquidation constraint \eqref{liquidation constraint}. Under suitable assumptions on the model parameters it was shown in \cite{Graewe2018, Horst2018} that the value function is given by $V_0=v_0|x|^p$ and that the optimal trading strategy is given by $\xi_0^*(t,y,x)=\frac{v_0(t,y)^\beta}{\eta(y)^\beta}x$ where $\beta=\frac{1}{p-1}$ and where $v_0$ is the unique nonnegative viscosity solution of polynomial growth to the following PDE:
\begin{equation}\label{v_0}
\left\{\begin{aligned}	&{-\partial_t v}(t,y)-\mathcal L v(t,y)- F(y,v(t,y))=0,    & (t,y)\in[0,T)\times\mathbb R^d,&\\
&\lim_{t\rightarrow T}v(t,y)=+\infty  & \text{locally uniformly on } \mathbb R^d&
\end{aligned}\right.
\end{equation}
where 
\[
	F(y,v):= \lambda(y)-\frac{|v|^{\beta+1}}{\beta\eta(y)^{\beta}}.
\]

\subsection{The liquidation model under uncertainty}
In order to analyse the impact of factor uncertainty on optimal liquidation strategies  we introduce the class $\mathcal Q$ of all probablity measures $Q$ whose density with respect to the benchmark measure $\mathbb P$ is given by
$$\frac{dQ}{d\mathbb P}=\mathcal E\left(\int\vartheta_sdW_s\right)_T,\quad Q\textit{-a.s.}$$
for some progressive process $\vartheta.$  Here, $\mathcal E(M)_t=\exp(M_t-\frac{\langle M \rangle_t}{2})$ denotes the Doleans-Dade exponential of a continuous semimartingale $M$. 
Thus, $Q	\ll \mathbb P $ for every probability measure $Q\in \mathcal Q$  and it follows from\cite[Lemma 3.1]{HernandezSchied2007} that
$$\int^T_0 |\vartheta_s|^2ds<\infty,\quad Q\textit{-a.s.}.$$

Since our focus is on the impact of uncertainty about the factor dynamics on the optimal trading rules, we assume that the Brownian motions $B$ and $W$ are independent. In this case  the unaffected price process is still a square-integrable martingale under every probability $Q \in \mathcal Q.$ In view of \eqref{cost}, we thus obtain the same form for the cost function for every given probability $Q$ in the set $\mathcal Q:$
\begin{align*}
	J_{Q}(t,y,x,\xi)&=\mathbb E_{Q}  \Big[\int_t^T\big( \eta(Y^{t,y}_s) |\xi_s|^p+\lambda(Y^{t,y}_s)|X_s|^p\big)\,ds\Big].
	\end{align*}
	
Following a standard approach in optimal decision making under model uncertainty introduced by Hansen and Sargent \cite{Hansen2001}, we do not restrict the set of measures {\it a priori} but add a penalty term to the objective function. Specifically, every probability measure $Q \in \mathcal Q$ receives a penalty
$$\Upsilon(Q):=\mathbb E_{Q}\left[\int^T_t \frac{1}{
\hat \theta_s}|\vartheta_s|^m ds\right].$$
The nonnegative process $\hat\theta = (\hat \theta_s) $ measures the degree of confidence in the reference model: the larger the process, the less deviations from the reference model are penalised. The case $\hat\theta_s\equiv 0$ corresponds to the benchmark model without factor uncertainty. The case $\hat \theta_s\equiv\hat \theta$ and $m=2$ corresponds to the entropic penalty function, see, e.g. \cite{Anderson2003, Bordigoni2007}.

To the best of our knowledge, Maenhout \cite{Maenhout2004} was the first to propose a state-dependent parameter $\hat \theta$ when considering the robust portfolio optimization problem of a power-utility investor. He considered an uncertainty-tolerance parameter of the $\hat \theta_s=\frac{\theta}{\mathcal W^{1-r}_s}$ where $\theta$ is a positive constant, $\mathcal W_s$ denotes the wealth of the investor at time $s$ and $r \in (0,1)$ denotes the exponent in the power utility function. This choice of $\hat\theta$ essentially corresponds to scaling the uncertainty-tolerance parameter by the value function. In his model, this leads to a solution that is invariant to the scale of wealth and is amenable to a rigorous mathematical analysis. 
Among other things, he found that for this choice of homothetic preferences the optimal solution under model uncertainty is observationally equivalent to the optimal solution without model uncertainty but increased risk aversion. 

%

In our context, the approach of Maenhout \cite{Maenhout2004} corresponds to the choice 
$$\hat \theta_s:=\frac{\theta}{a|X_s^{\xi}|^p}$$
and thus to the penalty functional
 $$\Upsilon(Q):=\mathbb E_{Q}\left[\int^T_t \frac{1}{
\theta}a|\vartheta_s|^m|X_s^{\xi}|^p ds\right],$$
where the constant $a:=\frac{(m-1)^{m-1}}{m^m}$ is chosen for analytical convenience. We thus model the costs associated with an admissible trading strategy $\xi$ and probability measure $Q\in\mathcal Q$ by
	\begin{equation*}
	\tilde J(t,y,x;\xi,\vartheta):=\mathbb E_{Q}\left[\int_t^T\left(\eta(Y_s^{t,y}) |\xi_s|^p+\lambda(Y_s^{t,y})|X_s^{\xi}|^p-\frac{1}{
\theta}a|\vartheta_s|^m|X_s^{\xi}|^p\right)ds\right]
\end{equation*}
define the value function of the stochastic control problem for each initial state $(t,y,x)\in[0,T)\times\mathbb R^{d}\times\mathbb R$ as
\begin{equation} \label{value-function}
	V(t,y,x):=\inf_{\xi\in\mathcal A(t,x)}\sup_{Q\in \mathcal Q }\tilde J(t,y,x;\xi,\vartheta).
\end{equation}

We asume throughout that $p>1, m\geq 2$. Before presenting the main results, we list our assumptions on the model parameters in terms of some positive constants $\underline{c}, \bar C$.
\begin{assumption}
(on the diffusion coefficients) 
\begin{itemize}
	\item[(L.1)] The drift function $b:\mathbb R^d\rightarrow \mathbb R^d$ is Lipschitz continuous and of linear growth, i.e. for each $y \in\mathbb R^d,$ $$|b(x)-b(y)|\leq \bar C|x-y|, \quad |b(y)|\leq \bar C(1+|y|).$$
	\item[(L.2)] The volatility function $\sigma:\mathbb R^d\rightarrow\mathbb R^{d\times n}$ is Lipschitz continuous and of linear growth, i.e. for each $y \in\mathbb R^d,$ 
	$$|\sigma(x)-\sigma(y)|\leq \bar C|x-y|, \quad |\sigma(y)|\leq \bar C (1+|y|).$$
		\item[(L.3)] The volatility function $\sigma$ is uniformly bounded by $\bar C$. 
		\item[(L.4)] The drift and volatility functions $b,\sigma$ belong to $C^1$ and $\sigma\sigma^\ast$ is uniformly positive definite.
\end{itemize}
\end{assumption}
\begin{assumption}
(on the cost coefficients and model parameters)
\begin{itemize}
	\item[(F.1)]The coefficients $\eta,\lambda,1/\eta:\mathbb R^d\rightarrow \mathbb [0,\infty)$ are continuous. Moreover, there exists constants $k_0\in(0, 1]$ such that for $y\in\mathbb R^d$,
	$$\lambda(y)\leq \bar C\langle y\rangle^{(1-k_0)m}$$
	and
	$$\underline{c}\,\langle y\rangle^{(1-pk_0)m}\leq \eta(y)\leq \bar C\langle y\rangle^{(1-k_0)m}.$$
	Let $n:=(1-k_0)m.$
	\item[(F.2)] The function $\eta$ is twice continuously differentiable, and $\|\frac{\mathcal L \eta }{\eta}\| \leq \bar C, \Big\|\frac{|D\eta|^{\alpha+1}}{\eta}\Big\|\leq \bar C$ where $$\mathcal L:=\frac{1}{2}\textit{tr}(\sigma\sigma^\ast D^2)+\left\langle b,D\right\rangle,\quad \alpha:=\frac{1}{m-1}.$$
	\item[(F.3)] The function $\lambda$ belongs to $C^1_b(\mathbb R^d)$ and $0<\underline{c}\leq\eta\leq \bar C.$ 
\end{itemize}
\end{assumption}

The assumptions on the diffusion coefficients are standard. Assumption (F.1) states that $\lambda$ is of polynomoial growth and that $\eta$ can be bounded from below and above by polynomial growth functions, whose order may be negative. Conditions similar to (F.2) and (F.3) have also been made in \cite{Horst2018} and \cite{Graewe2018}, respectively. 

\vspace{2mm}

\begin{example}\label{ex1}
The assumptions on the diffusion coefficients are satisfied for the two-dimensional diffusion process $Y=(Y^1, Y^2)$ given by 
\[
	dY^1_t = - Y^1_t dt + dW^1_t \quad \mbox{and} \quad dY^2_t =\mu dt+\sigma dW^2_t.
\]
The Ornstein-Uhlenbeck process $Y^1$ drives the market impact term while the arithmetic Browninan motion $Y^2$ drives the market risk. Specifically, if we chose $\eta= \tanh (-Y^1)+2$, then this process can be viewed as describing a stochastic liquidity process that fluctuates around a stationary level. Moreover, for the stochastic volatility model 
\[
	d\tilde P_t=\tilde{\sigma}(Y^2_t)dB_t
\]
for the reference price process the instantaneous volatility of the portfolio process is given by $\tilde{\sigma}^2(Y^2_t)|X_t|^2$. Hence, if $\tilde{\sigma}$ is bounded and continuously differentiable with bounded derivative, then $\lambda:=\tilde{\sigma}^2$ satisfies the preceding assumptions.   
%

\end{example}
\subsection{The main results}

If all the processes $\vartheta$ take values in a compact set $\Theta$ then all probability measures $Q$ in $\mathcal Q$ are equivalent to $\mathbb P$. 
In this case, the dynamic programming principle suggests that the value function satisfies the following Hamilton-Jacobi-Bellman-Issacs equation, cf. \cite[Theorem 2.6]{Fleming1989}
\begin{equation} \label{hjb-0}
-\partial_t V(t,y,x)-\mathcal L V(t,y,x)-\inf_{\xi\in\mathbb R}\sup_{\vartheta\in\Theta}	\mathcal H(t,y,x,\xi,\vartheta,V)=0, \quad (t,y,x)\in[0,T)\times\mathbb R^d\times\mathbb R,
\end{equation}
where 
$\mathcal H$ is given by
\begin{equation*}
\begin{split}
\mathcal H(t,y,x,\xi,\vartheta,V):=&\left\langle \sigma\vartheta,\partial_y V(t,y,x)\right\rangle-\xi \partial_xV(t,y,x)+c(y,x,\xi)-\frac{1}{\theta}a|\vartheta|^m|x|^p,
\end{split}
\end{equation*}
and
\[ 
\quad c(y,x,\xi):=\eta(y)|\xi|^p+\lambda(y)|x|^p.
\]

In our case the set of probability measures is not restricted {\it a priori}. This suggests to characterise the value function \eqref{value-function} in terms of the solution to the modified HJBI equation 
\begin{equation} \label{hjb}
-\partial_t V(t,y,x)-\mathcal L V(t,y,x)-\inf_{\xi\in\mathbb R}\sup_{\vartheta\in\mathbb R^d}	\mathcal H(t,y,x,\xi,\vartheta,V)=0, \quad (t,y,x)\in[0,T)\times\mathbb R^d\times\mathbb R.
\end{equation}
Since the function $\mathcal H$ separates additively into two terms that depend on  $\vartheta$ only and into two terms that depend $\xi$ only, 
\begin{equation*}
\begin{split}
\inf_{\xi\in\mathbb R}\sup_{\vartheta\in\mathbb R^d}	\mathcal H(t,y,x,\xi,\vartheta,V)=& \sup_{\vartheta\in\mathbb R^d}\{\left\langle \sigma\vartheta,\partial_yV(t,y,x)\right\rangle-\frac{1}{\theta}a|\vartheta|^m|x|^p\}\\
&+\inf_{\xi\in\mathbb R}\{-\xi \partial_xV(t,y,x)+c(y,x,\xi)\}.
\end{split}
\end{equation*}
The structure of cost function suggests an ansatz of the form $V(t,y,x)=v(t,y)|x|^p$. In this case, 
\begin{equation}\label{optimal-control-vartheta}
\begin{split}
\vartheta^*(t,y):=&\argmax_{\vartheta\in\mathbb R^d}\left\{\Big\langle \sigma\vartheta,Dv(t,y)\right\rangle-\frac{1}{\theta} a|\vartheta|^m\Big\}\\
=&\theta^\alpha(1+\alpha)|\sigma^\ast (y)Dv(t,y)|^{\alpha-1}\sigma^\ast(y)Dv(t,y),
\end{split}
\end{equation}
and
\begin{equation} \label{optimal-control-xi} 
\begin{split}
\xi^*(t,y):=&\argmin_{\xi\in\mathbb R}\Big\{ -p\xi v(t,y)|x|^{p-1}\sgn(x) +\eta(y)|\xi|^p\Big\}\\
=&\frac{v(t,y)^{\beta}}{\eta(y)^{\beta}}x ,
\end{split}
\end{equation}
where $\alpha=\frac{1}{m-1}, \beta=\frac{1}{p-1}$. Thus,   
\begin{equation*}
\begin{split}
\inf_{\xi\in\mathbb R}\sup_{\vartheta\in\mathbb R^d}	\mathcal H(t,y,x,\xi,\vartheta,V)=&\Big(H(y,Dv(t,y)) +F(y,v(t,y))\Big)x^p
\end{split}
\end{equation*}
where
\begin{equation} \label{nonlinearity}
F(y,v):= \lambda(y)-\frac{|v|^{\beta+1}}{\beta\eta(y)^{\beta}},\quad H(y,q):=\theta^\alpha|\sigma^\ast(y) q|^{\alpha+1}.
\end{equation}

Similarly to the discussion in \cite[Section 2.2]{Graewe2018}, we expect the value function to be charaterised by the following terminal value problem: 
\begin{equation} \label{pde-sup}
\left\{\begin{aligned}	&{-\partial_t v}(t,y)-\mathcal L v(t,y)-H(y,Dv(t,y))- F(y,v(t,y))=0,    & (t,y)\in[0,T)\times\mathbb R^d,&\\
&\lim_{t\rightarrow T}v(t,y)=+\infty  & \text{locally uniformly on } \mathbb R^d.&
\end{aligned}\right.
\end{equation}
The problem reduces to the terminal value problem \eqref{v_0} in the absence of model uncertainty $(H=0)$. The following theorem guarantees the existence of a unique nonnegative viscosity solution to this singular problem under conditions (L.1)-(L.3), (F.1), (F.2) and $\beta>\alpha$. The additional assumption $\beta>\alpha$ can also be found in \cite{Galise2016} where the authors study the entire solutions of a similar kind of elliptic equation. The proof is given in Section \ref{viscosity solution}. 

\begin{theorem}\label{thm-viscosity}
Let $\beta>\alpha.$ Under Assumptions (L.1)-(L.3), (F.1) and (F.2), the singular terminal value problem \eqref{pde-sup} admits a unique nonnegative viscosity solution $v$ in $$C_n([0,T^-]\times\mathbb R^d),$$
where $n$ is introduced in condition (F.1).
\end{theorem}

Since the maximizer $\vartheta^*$ in \eqref{optimal-control-vartheta} depends on $Dv$, we expect the verification theorem to require the candidate value function $v$ to be of class $C^{0,1}.$ As shown by the following theorem this can be guaranteed under additional assumptions on the model parameters. Specifically, we show that uniformly in $y$ as $t\rightarrow T$ the function $v$ satisfies 
\begin{equation*} 
	(T-t)^{1/\beta} v(t,y)= \eta(y)+O((T-t)^{1-\alpha/\beta}),  
\end{equation*}
and 
\begin{equation*} 
	 (T-t)^{1/\beta} Dv(t,y)= D\eta(y)+O((T-t)^{\frac{1}{2}-\alpha/\beta}).
\end{equation*}
Thus, under the additional assumption that $\beta>2\alpha$, we obtain the convergence of both the rescaled function $v$ and its rescaled derivative to market impact term, respectively its derivative at the terminal time: 
\begin{equation*} 
	\lim\limits_{t\rightarrow T}(T-t)^{1/\beta} v(t,y)= \eta(y),\quad  \lim\limits_{t\rightarrow T}(T-t)^{1/\beta} Dv(t,y)=D \eta(y).
\end{equation*}
 The proof of the following theorem is given in Section \ref{regularity}.

\begin{theorem}\label{thm-regularity}
Let $\beta>2\alpha.$ Under Assumptions (L.1)-(L.4), (F.1)-(F.3), the unique nonnegative viscosity solution $v$ in $C_b([0,T^-]\times\mathbb R^d)$ to the singular terminal value problem \eqref{pde-sup} belongs to $C^{0,1}([0,T)\times\mathbb R^d).$
\end{theorem}
The previously established regularity of the candidate value function is indeed enough to carry out the verification argument, which is proven in Section \ref{verification}. 

\begin{theorem} \label{thm-verifcation}
Let $\beta>2\alpha.$ Under Assumptions (L.1)-(L.4), (F.1)-(F.3), let $v\in C^{0,1}([0,T)\times\mathbb R^d)$ be the nonnegative viscosity solution to the singular terminal value problem \eqref{pde-sup}.  Then, the value function of the control problem \eqref{value-function} is given by $V(t,y,x)=v(t,y)|x|^p$, and the optimal control $(\xi^*,\vartheta^*)$ is given in feedback form by
	\begin{equation} \label{optimal-control}
	\xi_s^*= \frac{v(s,Y_s^{t,y})^\beta}{\eta(Y_s^{t,y})^\beta}X_s^* \quad \text{ and } \quad \vartheta^*_s=\theta^\alpha(1+\alpha)|\sigma^\ast(Y_s^{t,y})Dv(s,Y_s^{t,y})|^{\alpha-1}\sigma^\ast(Y_s^{t,y})Dv(s,Y_s^{t,y}).
	\end{equation}
	In particular, the resulting optimal portfolio process $(X^*_s)_{s\in[t,T]}$ is given by
	\begin{equation} \label{optimal-position-process}	
	X_s^*=x\exp\left(-\int_t^s\frac{v(r,Y_r^{t,y})^\beta}{\eta(Y_r^{t,y})^\beta}\,dr\right).
	\end{equation}
\end{theorem}

\begin{remark}
The preceding results shows  that -- as in \cite{Maenhout2004} -- the model with factor uncertainty is equivalent to the benchmark model \eqref{cost} when the market risk factor $\lambda$ is replaced $\lambda^{H}:=\lambda+H(y,Dv(t,y))$. In particular, under model uncertainty the investor liquidates the asset at a faster rate. 
\end{remark}

Ou final results provides a first order approximation of the value for the model with uncertainty in terms of the solution to the benchmark model without uncertainty when the investor is ``almost certain'' about the reference model.  

\begin{theorem} \label{thm-asymptotic}
Let $\beta>2\alpha.$ Let $w=v(T-t)^{1/\beta}$ and $w_0=v_0(T-t)^{1/\beta}$ where $v_0$ denotes the solution to the benchmark model. Under Assumptions (L.1)-(L.4), (F.1)-(F.3),  we have that on $[0,T]\times\mathbb R^d,$
\begin{equation}\label{asymptotic-eq}
\lim\limits_{\theta\rightarrow 0}\frac{w-w_0}{\theta^{\alpha}}=w_1
\end{equation}
where,  $w_1$ is a unique nonnegative solution to the following PDE:
\begin{equation}\label{w_1}
\left\{\begin{aligned}	&{-\partial_t v}(t,y)-\mathcal L v(t,y)-f_1(t,y,v(t,y))=0,    & (t,y)\in[0,T)\times\mathbb R,&\\
& v(T,y)=0,  & y\in \mathbb R^d.&
\end{aligned}\right.
\end{equation}
whose driver $$f_1(t,y,v)=|\sigma Dv_0|^{1+\alpha}(T-t)^{1/\beta}-\frac{(\beta+1) v_0^\beta}{\beta\eta^\beta}v+\frac{1}{\beta}\frac{v}{(T-t)}$$
depends on the solution to the benchmark model without factor uncertainty.  
\end{theorem}

\section{Viscosity solution }\label{viscosity solution}
In this section, we prove Theorem \ref{thm-viscosity}. The proof uses modifications of arguments given in \cite{Horst2018}. In a first step, we establish a comparison principle for semicontinuous viscosity solutions to (\ref{pde-sup}). Due to the terminal state constraint we cannot follow the usual approach of showing that if a l.s.c. supersolution dominates an u.s.c. subsolution at the boundary, then it also dominates the subsolution on the entire domain. Instead, we prove that if some form of asymptotic dominance holds at the terminal time, then it holds near the terminal time. 

In a second step, we construct a smooth sub- and a supersolution  to \eqref{pde-sup} satisfying the required assumptions. Using Perron's method, we can then establish the existence of an upper semi-continuous subsolution and of a lower semi-continuous supersolution, which are bounded by the respective smooth solutions. In particular, the semi-continuous solutions can be applied to the comparison principle. This establishes the existence of the desired continuous solution. 

We start with the following comparison principle. The proof is given in Section \ref{proof-comparison}. We emphasise that the comparison principle will only be used to prove the existence of a viscosity solution. This justifies the rather strong assumptions \eqref{asympotic} and \eqref{interval} below. 

\begin{proposition} \label{comparison}
	Assume that Assumptions (L.1)-(L.3), (F.1) and (F.2) hold. Let $n$ be as in condition (F.1). Fix $\delta\in(0,T].$  Let $\overline u \in LSC_n([T-\delta,T^-]\times\mathbb R^d)$ and $\underline u \in USC_n([T-\delta,T^-]\times\mathbb R^d)$ be a 	nonnegative viscosity super- and a viscosity subsolution to \eqref{pde-sup}, respectively.  If, uniformly on $\mathbb R^d$,  
	\begin{equation}\label{asympotic}
	\limsup\limits_{t \rightarrow T}\frac{\underline u(t,y)(T-t)^{1/\beta}-\eta(y)}{\langle y\rangle^n}\leq 0\leq\liminf\limits_{t\rightarrow T}\frac{\overline u(t,y)(T-t)^{1/\beta}-\eta(y)}{\langle y\rangle^n},
	\end{equation}
	and 
	\begin{equation}\label{interval}
	\sqrt[\beta]{\frac{\frac{1}{2}\beta+1}{\beta+1}}\eta(y)\leq \underline u(t,y)(T-t)^{1/\beta}, \quad\overline u(t,y)(T-t)^{1/\beta}\leq  C\langle y\rangle^n,\quad t\in[T-\delta, T),
	\end{equation}
	for a constant $C,$ then $$\underline u \leq \overline u\quad \text{on} \quad [T-\delta,T)\times\mathbb R^d.$$
\end{proposition}

We are now going to construct smooth sub- and supersolutions to \eqref{pde-sup} that satisfy the conditions \eqref{asympotic} and \eqref{interval} of the above proposition. 
The supersolution will be defined in terms of the function $$\hat h(t,y):=e^{L(T-t)} \langle y\rangle^n$$ where $n$ is introduced in condition (F.1), and where the constant $L$ will be determined later. Using the condition (F.1), we can find a constant $C_0>0$ such that 
\begin{equation*}
\begin{split}
&-\partial_t \hat h(t,y)-\mathcal{L}\hat h(t,y)-2^{\alpha}\bar C^{\alpha+1}|D{\hat h}(t,y)|^{\alpha+1}-\lambda(y)+\frac{\hat h(t,y)^{\beta+1}}{\beta\eta(y)^\beta}\\
&\geq L\hat h(t,y)-C_0\hat h(t,y)-C_0e^{\alpha L(T-t)}\hat h(t,y)-C_0\hat h(t,y)+C_0e^{\beta L(T-t)}\hat h(t,y)\\
&\geq (L-2C_0)\hat h(t,y)+C_0e^{\alpha L(T-t)}\hat h(t,y)(e^{(\beta-\alpha) L(T-t)}-1).
\end{split}
\end{equation*}
Choosing $L$ large enough, we get that
\begin{equation}\label{hat h}
 -\partial_t \hat h(t,y)-\mathcal{L}\hat h(t,y)-2^{\alpha}\bar C^{\alpha+1}|D{\hat h}(t,y)|^{\alpha+1}-\lambda(y)+\frac{\hat h(t,y)^{\beta+1}}{\beta\eta(y)^\beta}\geq 0, (t,y)\in[0, T]\times\mathbb R^d.
\end{equation}

\begin{lemma}\label{lemma-solution}
Suppose that Assumptions (L.1)-(L.3), (F.1) and (F.2) hold. Let $\epsilon:=1-\alpha/\beta.$ There exist constants $K>0, \delta\in(0,T]$ such that 
\begin{equation*}
\check v(t,y):=\frac{\eta(y)-\eta(y)\|\frac{\mathcal L \eta }{\eta}\|(T-t)}{(T-t)^{1/\beta}}
\end{equation*}
and
\begin{equation*}
\hat v(t,y):=\frac{\eta(y)+\eta(y)K(T-t)^\epsilon}{(T-t)^{1/\beta}}+\hat h(t,y)
\end{equation*}
are a nonnegative classical sub- and supersolution to~\eqref{pde-sup} on $[T-\delta,T)\times\mathbb R^d,$ respectively. Furthermore, $\check v, \hat v$ satisfy the conditions \eqref{asympotic} and \eqref{interval}.
\end{lemma}
\begin{proof}
In veiw of (F.2), the quantity $\|\frac{\mathcal L \eta }{\eta}\|$ is well-defined and finite; hence $\delta_0 :=1/\|\frac{\mathcal L \eta }{\eta}\| > 0$. It has been shown in \cite{Horst2018} that $\check v$ is a subsolution to~\eqref{pde-sup} on $[T-\delta_0,T)\times\mathbb R^d$ when $H=0.$ Since $H$ is nonnegative, we know that $\check v$ is still a subsolution on $[T-\delta_0,T)\times\mathbb R^d.$ 
We now verify that $\hat v$ is a nonnegative classical  supersolution to~\eqref{pde-sup} on $[T-\delta_1,T)\times\mathbb R^d$ for small $\delta_1.$ To this end, we first obtain by a direct computation that
\begin{equation*}
\begin{aligned} 
-\partial_t\hat v(t,y)-\mathcal L \hat v(t,y)=&-\frac{\eta(y)+K(1-\beta\epsilon)\eta(y)(T-t)^\epsilon+\beta\mathcal L \eta(y)(T-t)\big(1+K(T-t)^{\epsilon}\big)}{\beta(T-t)^{(\beta+1)/\beta}}\\
&-\partial_t {\hat h}(t,y)-\mathcal{L}{\hat h}(t,y).
\end{aligned}
\end{equation*}
Assuming that $K\delta_1^{\epsilon}\leq 1$ and $\delta_1\leq1,$ we see that $K(T-t)^{\epsilon}\leq1$ and $(T-t)^{1-\epsilon}\leq1$ for $t\in [T-\delta_1, T)$. Thus,
\begin{equation}\label{derivative-v-hat}
\begin{aligned} 
-\partial_t\hat v(t,y)-\mathcal L \hat v(t,y)\geq&-\frac{\eta(y)+K(1-\beta\epsilon)\eta(y)(T-t)^\epsilon+2\beta\bar C\eta(y)(T-t)^\epsilon}{\beta(T-t)^{(\beta+1)/\beta}}\\
&-\partial_t {\hat h}(t,y)-\mathcal{L}{\hat h}(t,y).
\end{aligned}
\end{equation}
Recalling the definition of $H$ and $F$ in \eqref{nonlinearity},
\begin{equation}\label{H-v-hat}
\begin{aligned} 
-H(y,D\hat v(t,y))&\geq -2^{\alpha}\bar C^{\alpha+1}\frac{|D\eta|^{\alpha+1}[1+K(T-t)^\epsilon]^{\alpha+1}}{(T-t)^{(1+\alpha)/\beta}}-2^{\alpha}\bar C^{\alpha+1}|D{\hat h}(t,y)|^{\alpha+1}\\
&\geq -2^{\alpha}\bar C^{\alpha+1}\Big\|\frac{|D\eta|^{\alpha+1}}{\eta}\Big\|\eta(y)\frac{[1+K(T-t)^\epsilon]^{\alpha+1}}{(T-t)^{(1+\alpha)/\beta}}-2^{\alpha}\bar C^{\alpha+1}|D{\hat h}(t,y)|^{\alpha+1}\\
&\geq -2^{2\alpha+1}{\bar C}^{\alpha+2}\frac{\eta(y)}{(T-t)^{(1+\alpha)/\beta}}-2^{\alpha}\bar C^{\alpha+1}|D{\hat h}(t,y)|^{\alpha+1}.
\end{aligned}
\end{equation}
Applying Bernoulli's inequality in the form $(u+v+w)^{\beta+1}\geq u^{\beta+1}+(\beta+1)u^\beta v+w^{\beta+1}$ for $u,v,w\geq 0$ to the term $|\hat v(t,y)|^{\beta+1}$ in $F$, we obtain
\begin{equation} \label{nonlinearity-v-hat}
-F(y,\hat v(t,y))\geq-\lambda (y)+\frac{\eta(y)^{\beta+1}+(\beta+1)\eta(y)^{\beta}\eta(y)K(T-t)^\epsilon}{\beta\eta(y)^\beta (T-t)^{(\beta+1)/\beta}}+\frac{\hat h(t,y)^{\beta+1}}{\beta\eta(y)^\beta}.
\end{equation}
Hence, adding \eqref{derivative-v-hat}, \eqref{H-v-hat} and \eqref{nonlinearity-v-hat} and using \eqref{hat h} yields,
\begin{equation}\label{com-eq}
\begin{aligned} 
&-\partial_t\hat v(t,y)-\mathcal L \hat v(t,y)-H(y,D\hat v(t,y))-F(y,\hat v(t,y))\\
&\geq\eta(y)\frac{(1+\epsilon)K-2\bar C-2^{2\alpha+1}{\bar C}^{\alpha+2}}{(T-t)^{(1+\alpha)/\beta}}\\
&\quad-\partial_t \hat h(t,y)-\mathcal{L}\hat h(t,y)-2^{\alpha}\bar C^{\alpha+1}|D{\hat h}(t,y)|^{\alpha+1}-\lambda(y)+\frac{\hat h(t,y)^{\beta+1}}{\beta\eta(y)^\beta}\\
&\geq\eta(y)\frac{(1+\epsilon)K-2\bar C-2^{2\alpha+1}{\bar C}^{\alpha+2}}{(T-t)^{(1+\alpha)/\beta}}.
\end{aligned}
\end{equation}
 Choosing $K\geq \frac{2\bar C+2^{2\alpha+1}{\bar C}^{\alpha+2}}{1+\epsilon}$ and then 
 $\delta_1=\min\{1,\sqrt[\epsilon]{\frac{1}{K}}\}$, we conclude that
\begin{equation*}
-\partial_t\hat v(t,y)-\mathcal L \hat v(t,y)-H(y,D\hat v(t,y))-F(y,\hat v(t,y))\geq 0, \quad (t,y)\in[T-\delta_1,T)\times\mathbb R^d.
\end{equation*}
Next, we prove that $\check v, \hat v$ satisfy the asymptotic behaviour \eqref{asympotic} and \eqref{interval}. Recalling the definition of $\check v, \hat v$ and using the condition (F.1), we have
\begin{equation} \label{order}
\begin{aligned}
(T-t)^{1/\beta}\check v(t,y)&=\eta(y)+\langle y\rangle^n O(T-t), \quad &\text{uniformly in $y$ as $t\rightarrow T$.}&\\
(T-t)^{1/\beta}\hat v(t,y)&=\eta(y)+ \langle y\rangle^n O((T-t)^\epsilon), \quad &\text{uniformly in $y$ as $t\rightarrow T$.}&
\end{aligned}
\end{equation}
From this, we see that
\begin{equation} \label{limit-v}
\begin{aligned}
\lim\limits_{t \rightarrow T}\frac{\check v(t,y)(T-t)^{1/\beta}-\eta(y)}{\langle y\rangle^n}=\lim\limits_{t\rightarrow T}\frac{\hat v(t,y)(T-t)^{1/\beta}-\eta(y)}{\langle y\rangle^n}=0, \quad\text{ uniformly on } \mathbb R^d,
\end{aligned}
\end{equation}  
which verifies the condition \eqref{asympotic}. The upper bound in \eqref{interval} can be obtained using the condition (F.1) again. Moreover, for the lower bound in \eqref{interval}, choosing $\delta:=\min\{\delta_0( 1-\sqrt[\beta]{\frac{\frac{1}{2}\beta+1}{\beta+1}}),\delta_1\},$  we have that for all $(t,y)\in [T-\delta,T)\times \mathbb R^d,$
\begin{equation*}
\hat v(t,y)(T-t)^{1/\beta}\geq\check v(t,y)(T-t)^{1/\beta}=\eta(y)-\eta(y)\|\frac{\mathcal L \eta }{\eta}\|(T-t)\geq\sqrt[\beta]{\frac{\frac{1}{2}\beta+1}{\beta+1}}\eta(y).
\end{equation*}
\end{proof}

\begin{remark}
Due to the presence of the gradient term $H$, an additional term \eqref{H-v-hat} needs to be dominated and thus we make the choice that $\epsilon=1-\alpha/\beta$. If $H=0$, we can choose $\epsilon=1$ as in \cite{Horst2018}. 
\end{remark}

We are now ready to prove the existence result.
\begin{proof}[Proof of Theorem \ref{thm-viscosity}]
 In order to apply Perron's method, we set
$$\mathcal S=\{u | u\text{ is a subsolution of }\eqref{pde-sup} \text{ on } [T-\delta,T)\times\mathbb R^d\text{ and }u\leq \hat v\}.$$ Since $\check v\in \mathcal S,$  the set $\mathcal S$ is non-empty. Thus, the function
\[
v(t,y)=\sup \{u(t,y):u\in\mathcal S\}
\]
is well-defined, belongs to $USC_n([T-\delta,T^-]\times\mathbb R^d)$ and satisfies that $\check v\leq v$. Classical arguments\footnote{ 
	The standard Perron method of finding viscosity solutions for elliptic PDEs can be found in \cite{Crandall1992}. We refer to \cite [Appendix A] {Zhan1999} for the proof of this method for parabolic equations.} show that the upper semi-continuous envelope $v^*$ of $v$ is a viscosity subsolution to~\eqref{pde-sup}. From \cite [Lemma~A.2] {Zhan1999}, the lower semi-continuous envelope $v_*$ of $v$ is a viscosity supersolution to~\eqref{pde-sup}. Since $\check v\leq v_*\leq v^*\leq \hat v, $ we have  for all $(t,y)\in [T-\delta,T)\times\mathbb R^d$ that
\begin{equation*}\label{bound-v}
\sqrt[\beta]{\frac{\frac{1}{2}\beta+1}{\beta+1}}\eta(y)\leq v_*(t,y)(T-t)^{1/\beta}, v^*(t,y)(T-t)^{1/\beta}\leq C\langle y\rangle^n, 
\end{equation*}
and
\begin{align*}
\frac{\check v(t,y)(T-t)^{1/\beta}-\eta(y)}{\langle y\rangle^n}\leq \frac{v_*(t,y)(T-t)^{1/\beta}-\eta(y)}{\langle y\rangle^n} &\leq  \frac{v^*(t,y)(T-t)^{1/\beta}-\eta(y)}{\langle y\rangle^n} \\
&\leq  \frac{\hat v(t,y)(T-t)^{1/\beta}-\eta(y)}{\langle y\rangle^n}.
\end{align*}
Hence, it follows from \eqref{limit-v} that
\begin{equation} 
\begin{aligned}
\lim\limits_{t \rightarrow T} \frac{v_*(t,y)(T-t)^{1/\beta}-\eta(y)}{\langle y\rangle^n}=\lim\limits_{t \rightarrow T}\frac{v^*(t,y)(T-t)^{1/\beta}-\eta(y)}{\langle y\rangle^n}=0,\quad\text{ uniformly on } \mathbb R^d.
\end{aligned}
\end{equation}  

From our comparison principle [Proposition \ref{comparison}] we can thus conclude that $v^*\leq v_*\text{ on } [T-\delta,T)\times\mathbb R^d$ , which shows that $v$ is the desired viscosity solution to \eqref{pde-sup} that belongs to $C_n([T-\delta,T^-]\times\mathbb R^d)$. 

Next, we find a sub- and supersolution to \eqref{pde-sup} on $[0,T-\delta]\times\mathbb R^d$ with terminal value $v(T-\delta,\cdot)$ at $t=T-\delta.$ 
Obviously, $0$ is a subsoultion of \eqref{pde-sup}.
We now conjecture that there exists $\overline{K}>0$ such that $\overline w:=\overline{K}\eta+\hat h(t,y)$ is a viscosity supersolution to \eqref{pde-sup}. In fact, since $v\leq \hat v$ at $t=T-\delta$, we see that 
	\begin{equation*}\label{local-1}
	v(T-\delta,y)\leq \frac{\bar C}{\delta^{1/\beta}}\eta(y)+\hat h(T-\delta,y),\quad y\in \mathbb R^d.
	\end{equation*}
In view of the condition (F.2) and the inequality \eqref{hat h}, we have that
	\begin{equation*}
	\begin{aligned}
	&-\partial_t w(t,y)-\mathcal Lw(t,y)-H(y,Dw)-F(y,w(t,y))\\
	&\geq -\overline K\mathcal L\eta(y)-2^{\alpha}\bar C^{\alpha+1}\overline K^{\alpha+1}|D\eta|^{\alpha+1}+\frac{1}{\beta}\overline K^{\beta+1}\eta(y)\\
	&\quad  -\partial_t \hat h(t,y)-\mathcal{L}\hat h(t,y)-2^{\alpha}\bar C^{\alpha+1}|Dh(t,y)|^{\alpha+1}-\lambda(y)+\frac{\hat h(t,y)^{\beta+1}}{\beta\eta(y)^\beta}\\
	&\geq\eta(y)[\frac{1}{\beta}\overline K^{\beta+1}-\overline K\bar C-2^{\alpha}{\bar C}^{\alpha+2}\overline K^{\alpha+1}]\\
	&>0,
	\end{aligned}
	\end{equation*}
for $\overline K$ large enough.  Furthermore, ${\overline w}^{\beta+1}/\eta^{\beta}$ is of polynomial growth of order $m$. Combining the general comparison principle [Proposition \ref{comparison-general}] with Perron's method, we obtain a viscosity solution $v\in C_n([0,T-\delta]\times\mathbb R^d)$. Hence from the comparison principle for continuous viscosity solutions Lemma \ref{lemma-comparison principle}, we get a unique global viscosity solution $v\in C_n([0,T^-]\times\mathbb R^d)$.  
\end{proof}
\section{Regularity of the viscosity solution}\label{regularity}
In Section \ref{viscosity solution}, we established the existence of a continuous viscosity solution $v$ to \eqref{pde-sup}. Unlike in \cite{Horst2018}, continuity is not enough to carry out our verification argument [Theorem \ref{thm-verifcation}], due to the dependence of the candidate value function on the gradient. In view of \eqref{optimal-control}, the candidate value function, i.e. the viscosity solution should be at least of class $C^{0,1}.$ To this end, we proceed as follows. First, we establish the existence of a solution of class $C^{0,1}$ to a modified PDE where the singularity is moved into the nonlinearity. This will provide us  with both the necessary regularity properties of the viscosity solution and {\it a priori} estimates of the solution and its gradient {\sl near the terminal time}. Subsequently, we use a standard link between FBSDEs and viscosity solutions, from which we can derive the differentiability of the viscosity solution on the {\sl whole time interval}.

\subsection{Mild solution}\label{mild solution}
In what follows, we assume that Assumptions (L.1)-(L.4) and (F.1)-(F.3) hold and that $\beta>2\alpha$.  Recalling the definition of $\epsilon$ in Lemma \ref{lemma-solution}, we know that $\epsilon=1-\frac{\alpha}{\beta}\in(\frac{1}{2},1)$. As dicussed before, the viscosity solution $v$ constructed in the previous section is of the form
\begin{equation} \label{straightforward-ansatz}
	 v(T-t,y)= \frac{\eta(y)+\tilde u(t,y)}{t^{1/\beta}},  
\end{equation}
for some function $\tilde u$ that satisfies
$$\tilde u(t,y)=  O(t^{\epsilon}) \text{ uniformly in $y$ as $t\rightarrow0$}.$$
We choose the following equivalent ansatz: 
\begin{equation} \label{educated-ansatz}
	\qquad\qquad\qquad v(T-t,y)= \frac{\eta(y)}{t^{1/\beta}}+\frac{u(t,y)}{t^{1+1/\beta}}, \quad u(t,y)=  O(t^{1+\epsilon}) \text{ uniformly in $y$ as $t\rightarrow0$}.
\end{equation}
It is worth pointing out that if $H = 0$, we can choose $\epsilon=1$ in \eqref{straightforward-ansatz} and \eqref{educated-ansatz}. Plugging the asymptotic ansatz into~\eqref{pde-sup} results in a semilinear parabolic equation for $u$ with finite initial condition. The proof of the following lemma is similar to \cite[Lemma 4.1]{Graewe2018} and hence omitted.   

\begin{lemma} \label{lemma-separation}
If, for some $\delta>0$, a function $u\in  C^{0,1}([0,\delta]\times\mathbb R^d)$ satisfies
\begin{equation}  \label{growth-of-u}
	|u(t,y)|\leq t\eta(y), \quad t\in[0,\delta],\, y\in\mathbb R^d,
\end{equation}
and solves the equation
\begin{equation} \label{pde-asymptotic}
\left\{\begin{aligned} 
	\partial_tu(t,y)&=\mathcal L u(t,y)+F_0(t,y,u(t,y),Du(t,y)), \quad &t\in (0,\delta]\,,y\in\mathbb R^d,&\\
	u(0,y)&=0, &y\in\mathbb R^d,&
\end{aligned}\right.
\end{equation}
where
\begin{equation*} 
\begin{split}
	F_0(t,y,u,Du)=&t\mathcal L \eta(y)+t^{p}\lambda(y)-\frac{\eta(y)}{\beta}\sum_{k=2}^\infty\dbinom{\beta +1}{k} \left        (\frac{u}{t\eta(y)}\right)^k \\
	& +\theta^\alpha t^{\epsilon}\left|\sigma^\ast(y)\left(\frac{Du}{t}+D\eta\right)\right|^{\alpha+1},
	\end{split}
\end{equation*}
then a local solution $v \in  C^{0,1}([T-\delta,T^-]\times\mathbb R^d) $ to problem~\eqref{pde-sup} is given by
\begin{equation*}
	v(t,y)= \frac{\eta(y)}{(T-t)^{1/\beta}}+\frac{u(T-t,y)}{(T-t)^{1+1/\beta}}.
\end{equation*}
\end{lemma}

The case where $H=0$ has been solved under additional regularity assumptions in \cite{Graewe2018} using an analytic semigroup approach. Due to the presence of $H$ in our case, we need to choose $\epsilon < 1$, which renders the analysis more complex.  In particular, the locally Lipschitz continuity in \cite[Lemma 4.5]{Graewe2018} no longer holds in our case. Instead, we solve equation \eqref{pde-asymptotic} using the weak continuous semigroup approach introduced in \cite[Section 4]{Fabbri2017} in order to obtain a $C^{0,1}$ solution.

In a first step we introduce the transition semigroup. Under Assumptions (L.1) and (L.2), the operator
\begin{equation*}
P_{t, s}[\varphi](y)=\mathbb E[\varphi(Y_s^{t,y})],\quad \varphi\in C_b(\mathbb R^d), 0\leq t\leq s
\end{equation*}
is well-defined and satisfies the Markov property $P_{t, r}=P_{t, s}P_{s, r}$ for $0\leq t\leq s\leq r.$ Since $b$ and $\sigma$ are independent of the time variable, 
$$P_{t, s}[\varphi](y)=P_{0, s-t}[\varphi](y).$$
For convenience, we denote
\begin{equation}
P_{t}[\varphi](y)=\mathbb E[\varphi(Y_t^{0,y})],\quad \varphi\in C_b(\mathbb R^d).
\end{equation}
For every $\varphi\in C_b(\mathbb R^d)$,
\begin{equation}\label{p-ass}
|P_{t}[\varphi](y)|\leq \|\varphi\|,\quad (t,y)\in[0,T]\times\mathbb R^d.
\end{equation}
Furthermore, from \cite[Theorem 4.65]{Fabbri2017}, we have the following proposition.
\begin{proposition}\label{prop-derivative}
Suppose that Assumptions (L.1)-(L.4) hold and let  $\varphi\in C_b(\mathbb R^d)$. Then for every $0\leq t\leq T,$ the function $y\rightarrow P_{t}[\varphi](y)$ is continuously differentiable on $ \mathbb R^d.$ Moreover,  there exists a constant $M>0$ such that for every $\varphi\in C_b(\mathbb R^d)$ and for $0\leq t\leq T,$
\begin{equation}\label{derivative-ass}
|DP_{t}[\varphi](y)|\leq \frac{M}{t^{1/2}}\|\varphi\|,\quad y\in\mathbb R^d.
\end{equation}
\end{proposition}

Next, we introduce the notion of a mild solution of our modified PDE. 

\begin{definition}\label{mildsolution}
We say that a function $u:[0,\delta]\times \mathbb R^d\rightarrow \mathbb R$ is a mild solution of the PDE \eqref{pde-asymptotic} if the following conditions are satisfied:
\begin{itemize}
	\item[(i)] $u \in C^{0,1}_b([0,\delta]\times\mathbb R^d).$
	\item[(ii)] for every $t\in[0,T]$ and $y\in \mathbb R^d,$
	\begin{equation}\label{mildform}
		u(t,y)=\int_0^t P_{t-s}[F_0(s,\cdot,u(s,\cdot),Du(s,\cdot))](y)ds.
	\end{equation}
\end{itemize}
\end{definition}

We prove the existence of a mild solution to \eqref{pde-asymptotic} by a contraction argument. To this end, we need to choose an apropriate weighted norm on $C^{0,1}_b([0,\delta]\times \mathbb R^d)$ to cope with the singularity in $F_0$. Recalling the ansatz \eqref{educated-ansatz} and the property \eqref{derivative-ass}, we consider the space
$$\Sigma:=\Big\{u\in C^{0,1}_b([0,\delta]\times \mathbb R^d): \|u(t,\cdot)\|+\|t^{1/2}Du(t,\cdot)\| =O(t^{1+\epsilon}) \textit{ as } t\rightarrow 0\Big\},$$
endowed with the weighted norm
$$\|u\|_\Sigma=\sup_{(t,y)\in (0,\delta]\times \mathbb R^d}\left(\frac{|u(t,y)|}{t^{1+\epsilon}}+\frac{|Du(t,y)|}{t^{1/2+\epsilon}}\right).$$
It is easy to verify that the vector space $\Sigma$ endowed with the norm $\|\cdot\|_\Sigma$ is a Banach space. 	
\begin{lemma}\label{lemma-f-continuous}
Suppose that $\beta>2\alpha$ and that Assumptions (L.1)-(L.4) and (F.1)-(F.3) hold. Let $R>0$ and $\delta\in (0,\sqrt[\epsilon-\frac{1}{2}]{\underline{c}/R}\wedge 1]$. Define the closed ball $\overline B_\Sigma(R):=\{u\in \Sigma: ||u||_\Sigma \leq R\}.$  For every $u\in \overline B_\Sigma(R),$ the function 
$$f_0(t,y):=F_0(t,y,u(t,y),Du(t,y))$$ is continuous.
\end{lemma}
\begin{proof}
For $u\in \overline B_\Sigma(R),$  we may decompose $f_0(t,y)$ in the following way:
\begin{equation} \label{decomposition}
	f_0(t,y)=t\mathcal L\eta(y)+t^{p}\lambda(y)-(p-1)\eta(y) g_0(t,y) +\theta^\alpha t^{\epsilon}g_1(t,y).
\end{equation}
where
\begin{equation*}
	g_0(t,y)=\sum_{k=2}^\infty \dbinom{\beta +1}{k}\left(\frac{u(t,y)}{t\eta(y)}\right)^k \quad \text{and} \quad 
g_1(t,y)=\left|\sigma^\ast(y)\left(\frac{Du(t,y)}{t}+D\eta(y)\right)\right|^{\alpha+1}.
\end{equation*}
The assumption $\delta\leq \sqrt[\epsilon-\frac{1}{2}]{\underline{c}/R}$ guarantees that the series  converges since then
\begin{equation*}
	\left|\frac{u(t,y)}{t\eta(y)}\right| \leq\frac{t^{1+\epsilon} R}{t\underline{c}}\leq \frac{\delta^\epsilon R}{\underline{c}}\leq 1, \quad t\in[0,\delta], y\in \mathbb R^d.
\end{equation*}
Moreover,
\begin{equation}\label{Du-esti}
	\left|\frac{Du(t,y)}{t}\right| \leq\frac{t^{\frac{1}{2}+\epsilon} R}{t}\leq  \delta^{\epsilon-\frac{1}{2}} R\leq \underline{c}, \quad t\in[0,\delta], y\in \mathbb R^d.
\end{equation}
In view of \eqref{decomposition} it is sufficient to prove that $g_0$ and $g_1$ are continuous in $t$, uniformly with respect to $y$ on every compact subset of $\mathbb R^d$. In fact, by the mean value theorem, we have for $0\leq t\leq s\leq\delta, y\in \mathbb R^d$ that
\begin{align*}
|g_1(t,y)-g_1(s,y)|&\leq\left|\left|\sigma^\ast(y)\left(\frac{Du(t,y)}{t}+D\eta(y)\right)\right|^{\alpha+1}-\left|\sigma^\ast(y)\left(\frac{Du(s,y)}{s}+D\eta(y)\right)\right|^{\alpha+1}\right|\\
&\leq (\alpha+1)\bar C^{\alpha+1}(\underline{c}+\bar C)^\alpha\left|\frac{Du(t,y)}{t}-\frac{Du(s,y)}{s}\right|.
\end{align*}
In order to establish the continuity of $g_0,$ notice that for every $k\geq 2$ and $0\leq t\leq s\leq\delta, y\in \mathbb R^d$ it holds that
\begin{equation} \label{first-estimate}
\begin{aligned}
	&\left|\left(\frac{u(t,y)}{t\eta(y)}\right)^k\right. -\left.\left(\frac{u(s,y)}{s\eta(y)}\right)^k\right|\\
&\leq \frac{1}{\underline{c}^k} \left|\frac{u(t,y)}{t}-\frac{u(s,y)}{s}\right|\sum_{l=0}^{k-1}\left|\frac{u(t,y)}{t}\right|^l\left|\frac{u(s,y)}{s}\right|^{k-1-l}
\\
&\leq \frac{R^{k-1}}{\underline{c}^k} \left|\frac{u(t,y)}{t}-\frac{u(s,y)}{s}\right|\sum_{l=0}^{k-1}t^{\epsilon l}s^{\epsilon(k-1-l)}\\
&\leq \frac{kR^{k-1}}{\underline{c}^k}\left|\frac{u(t,y)}{t}-\frac{u(s,y)}{s}\right| s^{(k-1)\epsilon}\\
&\leq \frac{k}{\underline{c}}(\frac{Rs^\epsilon}{\underline{c}})^{k-1}\left|\frac{u(t,y)}{t}-\frac{u(s,y)}{s}\right|.
\end{aligned}
\end{equation}
Using the identity $ k\binom{\beta+1}{k}=(\beta+1)\binom{\beta}{k-1}$, we get that
\begin{equation*}
|g_0(t,y)-g_0(s,y)|\leq (\beta+1)\max\{2^\beta-1,\beta\}\frac{Rs^\epsilon}{\underline{c}^2}\left|\frac{u(t,y)}{t}-\frac{u(s,y)}{s}\right|.
\end{equation*}
Hence the claim follows from the fact that the maps $(t,y)\mapsto\frac{u(t,y)}{t},\frac{Du(t,y)}{t}$ are continuous on $[0,\delta]\times \mathbb R^d$.
\end{proof}

The following lemma can be established using similar arguments as above.

\begin{lemma} \label{lemma-locally-lip}
Suppose that $\beta>2\alpha$ and that Assumptions (L.1)-(L.4) and (F.1)-(F.3) hold.  For every $R>0$ there exists a constant $L>0$ independent of $\delta\in (0,\sqrt[\epsilon-\frac{1}{2}]{\underline{c}/R}]$ such that
\begin{align*}
	&\left|F_0(t,y,u(t,y),Du(t,y))-F_0(t,y,v(t,y),Dv(t,y))\right|\\
	\leq &Lt^{\epsilon}\left(\frac{\left|u(t,y)-v(t,y)\right|}{t}+\frac{\left|Du(t,y)-Dv(t,y)\right|}{t}\right), \quad   u,v\in\overline B_\Sigma(R),\, \, t\in[0,\delta], y\in \mathbb R^d.
\end{align*}
\end{lemma}
We are now ready to carry out the fixed point argument.
\begin{theorem}\label{thm-mild-local}
Let $\beta>2\alpha$. Under Assumptions (L.1)-(L.4) and (F.1)-(F.3), there exists a constant $\delta>0$ such that Equation \eqref{pde-asymptotic}  admits a mild solution $u \in C^{0,1}_b([0, \delta]\times\mathbb R^d).$ 
\end{theorem}
\begin{proof}
Let us define the operator 
\begin{equation}\label{operator}
\Gamma[u](t,y):=\int^t_0 P_{t-s}[F_0(s,\cdot,u(s,\cdot),Du(s,\cdot))](y)ds
\end{equation}
\textsc{Step 1: the map $\Gamma$ is well defined on the closed ball $\overline B_\Sigma(R)$.}
Let $u\in \overline B_\Sigma(R).$  By Lemma \ref{lemma-f-continuous} and \cite[Proposition 4.67]{Fabbri2017}\footnote{The strong continuity in this proposition is equivalent to the standard continuity in finite-dimensional space.},  we see that $\Gamma[u]\in C_b([0, \delta]\times\mathbb R^d) $ and $D\Gamma[u]\in C_b((0, \delta]\times\mathbb R^d)$. In order to see the continuity of $D\Gamma[u]$ at $t=0,$ we differentiate \eqref{operator} to obtain that 
\begin{equation}\label{operator-D}
D\Gamma[u](t,y)=\int^t_0 DP_{t-s}[F_0(s,\cdot,u(s,\cdot),Du(s,\cdot))](y)ds, \quad (t,y)\in[0, \delta]\times\mathbb R^d.
\end{equation}
By Proposition \ref{prop-derivative},
\begin{equation*}
|D\Gamma[u](t,y)|\leq\int^t_0 M\frac{\|f_0\|}{(t-s)^{1/2}}ds=\sqrt{t}M\|f_0\|.
\end{equation*}
From this, we conclude that the map $(t,y)\mapsto D\Gamma[u](t,y)$ belongs to $C_b([0, \delta]\times\mathbb R^d)$. 

\textsc{Step 2: contraction property of $\Gamma$ on $\overline B_\Sigma(R)$ for a suitable choice of $R, \delta$.}
 Let  $$B(a,b):=\int^1_0r^{a-1}(1-r)^{b-1}dr$$ be the Beta function with $a, b>0$. We choose
\begin{equation*}
	R=2(1+MB_0)\left(\|\mathcal L\eta\|+\|\lambda\|+\|\sigma^\ast D\eta\|^{\alpha+1}\right),
\end{equation*}
and
\begin{equation*}
	 \delta=\min\{\sqrt[\epsilon-\frac{1}{2}]{\underline{c}/R},\sqrt[\epsilon-\frac{1}{2}]{1/\big(2L(1+MB_1)\big)},1\},
\end{equation*}
where $L > 0$ is the Lipschitz constant given by Lemma~\ref{lemma-locally-lip} and $$B_0:=B(1+\epsilon,\frac{1}{2}),\quad B_1:=B(2\epsilon+\frac{1}{2},\frac{1}{2}).$$

 Let $u, v\in \overline B_\Sigma(R).$  By Lemma \ref{lemma-locally-lip}, we have  for $(t,y)\in [0,\delta]\times\mathbb R^d$ that 
 \begin{equation*}
\begin{aligned}
&|\Gamma[u](t,y)-\Gamma[v](t,y)|\\
=&\left|\int^t_0 P_{t-s}[F_0(s,\cdot,u(s,\cdot),Du(s,\cdot))-F_0(s,\cdot,v(s,\cdot),Dv(s,\cdot))](y)ds\right|\\
\leq& \int^t_0 \left\|F_0(s,y,u(s,\cdot),Du(s,\cdot))-F_0(s,\cdot,v(s,\cdot),Dv(s,\cdot))\right\|ds\\
\leq &\int^t_0 Ls^{\epsilon}\left(\frac{\left\|u(s,\cdot)-v(s,\cdot)\right\|}{s}+\frac{\left\|Du(s,\cdot)-Dv(s,\cdot)\right\|}{s})\right)ds\\
=&\int^t_0 L\left(s^{2\epsilon}\frac{\left\|u(s,\cdot)-v(s,\cdot)\right\|}{s^{1+\epsilon}} +s^{2\epsilon-1/2}\frac{\left\|Du(s,\cdot)-Dv(s,\cdot)\right\|}{s^{1/2+\epsilon}}\right)ds\\
\leq& Lt^{2\epsilon+1/2}\|u-v\|_\Sigma.
\end{aligned}
\end{equation*}
Similarly,
 \begin{equation*}
\begin{aligned}
&|D\Gamma[u](t,y)-D\Gamma[v](t,y)|\\
=&\left|\int^t_0 DP_{t-s}[F_0(s,\cdot,u(s,\cdot),Du(s,\cdot))-F_0(s,\cdot,v(s,\cdot),Dv(s,\cdot))](y)ds\right|\\
\leq& M\int^t_0\frac{1}{(t-s)^{1/2}}\left\|F_0(s,y,u(s,\cdot),Du(s,\cdot))-F_0(s,\cdot,v(s,\cdot),Dv(s,\cdot))\right\|ds\\ 
\leq& \int^t_0 ML\frac{1}{(t-s)^{1/2}}\left(s^{2\epsilon-1/2}\|u-v\|_\Sigma\right)ds\\
\leq& MLB_1t^{2\epsilon}\|u-v\|_\Sigma.
\end{aligned}
\end{equation*}
Hence 
\begin{equation*}
\|\Gamma[u]-\Gamma[v]\|_\Sigma\leq \frac{1}{2	}\|u-v\|_\Sigma.
\end{equation*}

\textsc{Step 3:  $\Gamma$ maps $\overline B_\Sigma(R)$ into itself.}
Note that $s^k\leq 1$ for all $k>0$ and $s\in[0,\delta]$ since $\delta\leq 1$. Hence, it holds for every $t\in[0,\delta]$ that
\begin{align*}
|\Gamma[0](t,y)|&=\left|\int^t_0 P_{t-s}[F_0(s,\cdot,0,0)](y)ds\right|\\
&\leq  \int^t_0 \|s\mathcal L \eta+s^{p}\lambda+\theta^\alpha s^{\epsilon}|\sigma^\ast D\eta|^{\alpha+1}\|\,ds\\
&\leq t^{1+\epsilon}(\|\mathcal L \eta\|+\|\lambda\|+\|\sigma^\ast D\eta\|^{\alpha+1}\|),
\end{align*}
and
\begin{align*}
|D\Gamma[0](t,y)|&=\left|\int^t_0 DP_{t-s}[F_0(s,\cdot,0,0)](y)ds\right|\\
&\leq  \int^t_0 \frac{1}{(t-s)^{1/2}}M\|s\mathcal L \eta+s^{p}\lambda+\theta^\alpha s^{\epsilon}|\sigma^\ast D\eta|^{\alpha+1}\|\,ds\\
&\leq t^{1+\epsilon-1/2}MB_0(\|\mathcal L \eta\|+\|\lambda\|+\|\sigma^\ast D\eta\|^{\alpha+1}).
\end{align*}
Thus,
\begin{align*}
\|\Gamma[u]\|_\Sigma&\leq \|\Gamma[u]-\Gamma[0]\|_\Sigma+\|\Gamma[0]\|_\Sigma\leq R.
\end{align*}
Hence, $\Gamma$ is a contraction from $\overline B_\Sigma(R)$ to itself and has a unique fixed point $u$ in $\overline B_\Sigma(R)$.
\end{proof}
\subsection{Gradient estimate of the viscosity solution}
It can be easily proved that the mild solution $u\in C^{0,1}_b([0,\delta]\times\mathbb R^d)$ obtained in Theorem \ref{thm-mild-local} is also a viscosity solution of \eqref{pde-asymptotic} on $[0,\delta]\times\mathbb R^d$. Thus $$w(t,y):= \frac{\eta(y)}{(T-t)^{1/\beta}}+\frac{u(T-t,y)}{(T-t)^{1+1/\beta}}$$ is a viscosity solution of \eqref{pde-sup} in $ C^{0,1}_b([T-\delta, T^-]\times\mathbb R^d)$. By Lemma \ref{lemma-comparison principle},  $v=w$ on $[T-\delta, T)\times\mathbb R^d.$ In view of \eqref{Du-esti} and the boundedness of $D\eta$ derived from (F.2) and (F.3),  we see that there exits a constant $C>0$ such that for $(t,y)\in[T-\delta, T)\times\mathbb R^d,$
\begin{equation}\label{Dv-bound}
|Dv(t,y)|\leq \frac{C}{(T-t)^{1/\beta}}.
\end{equation}
It remains to establish an {\it a priori} estimate for $Dv$ on $[0,T-\delta]\times\mathbb R^d.$ To this end, we introduce a family of quadratic FBSDE systems whose terminal value at time $T_0\in (0,T)$ is given by $v(T_0,\cdot)$. The first component of the solution to the BSDE is given in terms of the viscosity solution. The differentiability of the viscosity solution can then be inferred from the differentiability of the corresponding BSDE. 
\begin{lemma}\label{lemma-BSDE}
Suppose that $\beta>2\alpha$ and that Assumptions (L.1)-(L.4) and (F.1)-(F.3) hold.  There exists processes $ (U^{t,y}, Z^{t,y})\in S^{\infty}_\mathcal{F}(t,T^-;\mathbb R^d)\times H^q_{\mathcal{F}}(t,T^-;\mathbb R^{1\times n})$ for all $q\geq2$ satisfying  $U^{t,y}_t=v(t,y)$ and for any $t\leq r\leq s<T,$
\begin{equation}\label{BSDE-whole}
U^{t,y}_r=U^{t,y}_s+\int^s_r F(Y^{t,y}_{\rho},U^{t,y}_{\rho})+\theta^\alpha|Z^{t,y}_{\rho} |^{1+\alpha}d\rho- \int^s_r Z^{t,y}_{\rho} dW_{\rho}.
\end{equation}
\end{lemma}
\begin{proof}
For $T_0\in (0,T),$ we conisder the PDE
\begin{equation}\label{originalpde}
\left\{
\begin{aligned}
&-\partial_t w(t,y)-\mathcal Lw(t,y)-H(Dw(t,y))-f(t,y)=0, (t,y)\in [0,T_0)\times \mathbb R^d,\\
&w(T_0, y)=v(T_0, y),  y\in \mathbb R^d.
\end{aligned}\right.
\end{equation}
where $f(t,y):=F(y,v(t,y))$ for $(t,y)\in[0,T_0] \times\mathbb R^d,$ and the forward-backward system
\begin{equation}\label{orignalbsde}
\left\{
\begin{aligned}
dY^{t,y}_s&=b(Y^{t,y}_s)ds+\sigma(Y^{t,y}_s)dW_s, \quad s\in[t,T_0]\\
dU^{t,y}_s&=-f(s,Y^{t,y}_s)-\theta^\alpha|Z^{t,y}_s|^{1+\alpha}ds+Z^{t,y}_sdW_s,\quad s\in[t,T_0]\\
Y^{t,y}_t&=y, U^{t,y}_{T_0}=v(T_0,Y^{t,y}_{T_0}).
\end{aligned}\right.
\end{equation}

From \cite[Theorem 1]{Imkeller2010}, the system \eqref{orignalbsde} admits a unique solution $(Y^{t,y}_s,U^{t,y}_s,Z^{t,y}_s)_{t\leq s\leq T_0}$ in the space $S^{2}_\mathcal{F}(t,T^-;\mathbb R^d)\times S^{\infty}_\mathcal{F}(t,T^-;\mathbb R^d)\times H^q_{\mathcal{F}}(t,T^-;\mathbb R^{n})$ and $\int^\cdot_t Z_sdW_s$ is a BMO martingale. Furthermore, the map $(t,y)\mapsto U^{t,y}_t$ defines a viscoisty solution of \eqref{originalpde} by \cite[Proposition 8]{Briand2007}. Hence it follows from the comparison principle [Proposition \ref{comparison-f}] that $U^{t,y}_t=v(t,y)$ for $(t,y)\in[0,T_0] \times\mathbb R^d.$ As a result, we have for any $r\in[t,T_0]$ that $0\leq U^{t,y}_r=v(r,Y^{t,y}_r).$ Thus $(U^{t,y}_s,Z^{t,y}_s)_{t\leq s\leq T_0}$ is also a solution to the following BSDE:
\begin{equation*}\label{BSDE'}
\left\{
\begin{aligned}
dU^{t,y}_s&=-F(Y^{t,y}_s, U^{t,y}_s)-\theta^\alpha|Z^{t,y}_s|^{1+\alpha}ds+Z^{t,y}_sdW_s,\quad s\in[t,T_0]\\
U^{t,y}_{T_0}&=v(T_0,Y^{t,y}_{T_0}).
\end{aligned}\right.
\end{equation*}
Since $T_0$ is arbitrary, we obtain a solution to the BSDE \eqref{BSDE-whole} on $[0,T)$. 
\end{proof}
\begin{proposition}\label{V-deri}
Let $\beta>2\alpha$. Under Assumptions (L.1)-(L.4) and (F.1)-(F.3), the function $v(t,\cdot)$ is continuously differentiable for any $t\in [0,T).$  
In addition, for every $y\in\mathbb R^d, \, 0\leq t\leq r<T,$
$$Z^{t,y}_r=Dv(r,Y^{t,y}_r)\sigma(Y^{t,y}_r), $$
where $Z^{t,y}$ is the second component of the solution to the BSDE \eqref{BSDE-whole}, and
\begin{equation}\label{Z-estimate}
|Z^{t,y}_r|\leq\left\{
\begin{aligned}
&\frac{C}{(T-r)^{1/\beta}},\quad r\in[T-\delta, T);\\
&C\left(1+\frac{1}{\delta^{1/\beta}}\right),\quad r\in[t, T-\delta].
\end{aligned}
\right.
\end{equation}
\end{proposition}
\begin{proof}
Since we have proved that $v(r,\cdot)$ is differentiable for $r\in [T-\delta,T)$, it follows by It\^o's formula that $Z^{t,y}_r=Dv(r,Y^{t,y}_r)\sigma(Y^{t,y}_r),$ for $r\in[T-\delta,T).$  The estimate on $[T-\delta,T)$ can thus be obtained from (L.3) and \eqref{Dv-bound}.

Next, we extend the domain of the solution by setting $Y^{t,y}_s = y$ for $s\in [0, t)$ and then consider the BSDE \eqref{BSDE-whole} on $[0,T-\delta]$. From \cite[Proposition 12]{Briand2008}, the map $(t,y)\mapsto (U^{t,y}_{\cdot},Z^{t,y}_{\cdot})$ belongs to $C^{0,1}([0,T-\delta]\times \mathbb R^d;S^{\infty}_\mathcal{F}\times H^q_{\mathcal{F}}).$  Moreover, by \cite[Theorem 15]{Briand2008}, for $0\leq t\leq r\leq T-\delta $, the map $y\mapsto U^{t,y}_{t}=v(t,y)$  is differentiable  and $Z^{t,y}_r=Dv(r,Y^{t,y}_r)\sigma(Y^{t,y}_r).$ 
The estimate on $Z^{t,y}$ can be obtained using the similar argument in the proof of\cite[Theorem 3.6]{Richou2012}. We sketch the proof for the reader's convenience. Denote the generator in \eqref{BSDE-whole} by $g$,  differentiating \eqref{BSDE-whole} yields
\begin{align*}
D U^{t,y}_r=&Dv(T-\delta,Y^{t,y}_{T-\delta})D Y^{t,y}_{T-\delta}-\int^{T-\delta}_r (D Z^{t,y}_{\rho})^\ast dW_{\rho}\\
&+\int^{T-\delta}_r  \partial_y g\cdot D Y^{t,y}_\rho+ \partial_u g\cdot D U^{t,y}_\rho+ \partial_z g\cdot D Z^{t,y}_{\rho}\,d\rho
\end{align*}
where 

\begin{align*}
&\partial_y g=D\lambda(Y^{t,y})+D \eta(Y^{t,y})\left(\frac{U^{t,y}}{\eta(Y^{t,y})}\right)^{\beta+1};\\
&\partial_u g=-\frac{\beta+1}{\beta}\left(\frac{U^{t,y}}{\eta(Y^{t,y})}\right)^{\beta};\\
&\partial_z g=(\alpha+1)\theta^\alpha|Z^{t,y}|^{\alpha-1}Z^{t,y},
\end{align*}

 and $Z^{t,y}_r=D U^{t,y}_r(D Y^{t,y}_r)^{-1}\sigma(Y^{t,y}_r)$. Furthermore, from \cite[Corollary 4.1]{Karoui1997}, 
since $\int^\cdot_t Z^{t,y}_\rho dW_{\rho}$ is BMO and
$$|\partial_z g|\leq C(1+|Z^{t,y}|),$$
the process $\int^\cdot_t\partial_z g dW_\rho$ is BMO. We can thus apply Girsanov's theorem to see that
\begin{equation*}
\tilde W_r=W_r-\int^r_t \partial_z g\,d\rho
\end{equation*}
is a Brownian motion under the probability 
\begin{equation*}
\frac{dQ}{d\mathbb P}=\mathcal E\left(\int^\cdot_t \partial_z g\,dW_\rho\right).
\end{equation*}
We obtain that
\begin{align*}
D U^{t,y}_r=&\mathbb E^Q\left[e^{\int^{T-\delta}_r\partial_u gd\rho }Dv(T-\delta,Y^{t,y}_{T-\delta})D Y^{t,y}_{T-\delta}+\int^{T-\delta}_r e^{\int^\rho_r\partial_u gd\rho }\partial_y g\cdot D Y^{t,y}_\rho\,d\rho\right]
\end{align*}
and hence
\begin{align}\label{Z_r}
|Z^{t,y}_r|\leq C\big(1+\frac{1}{\delta^{1/\beta}}\big)\cdot\mathbb E^Q\left[\sup_{r\leq \rho\leq T-\delta}|D Y^{t,y}_{\rho}(D Y^{t,y}_r)^{-1}|\right].
\end{align}
by the boundedness of $\partial_u g, \partial_y g$ and the estimate \eqref{Dv-bound}.
Let us denote
$$\mathcal E_{r,T-\delta}:=\mathcal E\left(\int^{T-\delta}_r \partial_z g\,dW_\rho\right).$$
Since $\int^\cdot_t Z_\rho dW_{\rho}$ is BMO, there exists $q>1$ such that $\mathbb E[\mathcal E_{r,T-\delta}^q]<+\infty.$ Moreover, $D Y_{\rho}(D Y_r)^{-1}$ solves the SDE
$$\tilde Y^{t,y}_{\rho}=I_d+\int^{\rho}_r Db(Y^{t,y}_\zeta)\tilde Y^{t,y}_{\zeta}d\zeta+\sum^n_{i=1}\int^{\rho}_r D\sigma^i(Y^{t,y}_\zeta)\tilde Y^{t,y}_{\zeta}dW^i_{\zeta}.$$ By classical SDE estimates, we have that
$$\mathbb E^Q\left[\sup_{r\leq \rho\leq T-\delta}|D Y^{t,y}_{\rho}(D Y^{t,y}_r)^{-1}|\right]\leq \mathbb E\left[\mathcal E_{r,T-\delta}^q\right]^{1/q}\mathbb E\left[\sup_{r\leq \rho\leq T-\delta}|D Y^{t,y}_{\rho}(D Y^{t,y}_r)^{-1}|^{q^{\prime}}\right]^{1/{q^{\prime}}}\leq C,$$
where $q^{\prime}$ is the conjugate of $q.$ Putting this inequality into \eqref{Z_r} completes the proof. 
\end{proof}
\section{Verification}\label{verification}
This section is devoted to the verification argument. We first prove admissibility of the strategy $\xi^*$ by using the estimates of the nonnegative viscosity solution $v$ derived from the proof of Theorem \ref{thm-viscosity}. Since the optimal density $\vartheta^*$ takes values in an unbounded set, one needs an additional argument to guarantee that the corresponding stochastic exponential is a true martingale. Subsequently, we show that $(\xi^*,\vartheta^*)$ is a saddle point of the cost function and is indeed optimal.
\begin{lemma}	\label{lemma-admissible}
	The feedback controls $\xi^*$ given by \eqref{optimal-control} is admissible, and the portfolio process $(X_s^*)_{s\in[t,T]}$ is monotone.
\end{lemma}
The proof is similar to of \cite[Lemma 3.8]{Horst2018} and hence omitted. The following lemma shows that for any $\xi\in \mathcal A(t,x)$ the expected residual costs vanish as $s\rightarrow T$ under a particular class of equivalent measure.	
\begin{lemma} \label{lemma-finite-costs}
	For every $\xi\in  \mathcal A(t,x)$ and every $Q\in \mathcal Q$ satisfying $$\mathbb E\left[e^{q\int^T_t |\vartheta_r|^2\,dr} \right]<\infty, \quad \text{ for every } q>0,$$ it holds that
\begin{equation} \label{vanish}
	\mathbb E_Q\left[v(s,Y_s^{t,y})|X_s^{\xi}|^p\right] \longrightarrow 0, \quad \text{$s\rightarrow T$.}
\end{equation}
\end{lemma}
\begin{proof}
Set $\pi_s=\mathcal E(\int^s_t \vartheta_rdW_r).$ For $k>1, s\in[t,T],$
\begin{equation*} 
	\mathbb E\left[(\pi_s)^k\right] =\mathbb E\left[\mathcal E(k\int^s_t \vartheta_rdW_r)\right] \cdot \mathbb E\left[e^{\frac{k^2-k}{2}\int^s_t |\vartheta_r|^2\,dr}\right] <\infty.
\end{equation*}
Using the similar argument as in \cite{Horst2018}, we obtain
\[
|X_s^{\xi}|^p\leq C (T-s)^{1/\beta}E\left[\int_s^T|\xi_r|^p\,dr\Big|\mathcal F_s\right],
\]
and
\[
v(s,Y_s^{t,y})\leq \frac{C}{(T-s)^{1/\beta}}.
\]
Therefore, 
\begin{align*} 
 \mathbb E_Q\left[v(s,Y_s^{t,y})|X_s^{\xi}|^p\right]&=\mathbb E\left[\pi_s v(s,Y_s^{t,y})|X_s^{\xi}|^p\right]\\
 &\leq C\mathbb E\left[ \pi_s\int_s^T|\xi_r|^p\,dr \right]\\
	&\leq C\left((T-s) \mathbb E\left[ (\pi_s )^2\right]\mathbb E\left[\int_s^T|\xi_r|^{2p}\,dr\right]\right)^{1/2}.
\end{align*}
Letting $s\rightarrow T$, the desired result \eqref{vanish} follows since $\xi\in L^{2p}_{\mathcal F}(0,T;\mathbb R)$.
\end{proof}
Now we are ready to carry out the verification argument. We will show that  $v(\cdot,\cdot)|\cdot|^p$ is indeed equal to the value function of our control problem and that the candidate strategy is optimal on the whole time interval. 

\begin{proof}[Proof of Theorem \ref{thm-verifcation}]
	For fixed $t\leq s<T,$  by Lemma \ref{lemma-BSDE} we have that
	\begin{equation*}
	\begin{split}
	U^{t,y}_t=& U^{t,y}_s+ \int_t^{s}\left( F(Y^{t,y}_{r},U^{t,y}_r)+|Z^{t,y}_r |^{1+\alpha}\right)\,dr - \int_t^{s}Z^{t,y}_r\,dW_r.
	\end{split}
	\end{equation*}
	This allows us to apply to $U^{t,y}_r|X_{r}^{\xi}|^p$ the integration by parts formula on $[t,s]$ and to get that 
	\begin{align*}
	U^{t,y}_t|x|^p =& U^{t,y}_s|X_{s}^{\xi}|^p +\int_t^{s}\big\{(F(Y^{t,y}_{r},U^{t,y}_r)+\theta^\alpha|Z^{t,y}_r |^{1+\alpha})|X_r^{\xi}|^p \nonumber\\
	&+p\xi_r U^{t,y}_r \sgn(X_r^{\xi})|X_r^{\xi}|^{p-1})\big\}\,dr -\int_t^{s}Z^{t,y}_r|X_r^{\xi}|^p\,dW_r .
	\end{align*}
Denote $W^{\vartheta}_r=W_r-\int^r_t \vartheta_\rho d\rho.$ Thus,
		\begin{align}\label{ito}
	U^{t,y}_t|x|^p =& U^{t,y}_s|X_{s}^{\xi}|^p +\int_t^{s}\big\{(F(Y^{t,y}_{r},U^{t,y}_r)+\theta^\alpha|Z^{t,y}_r |^{1+\alpha}-\vartheta_rZ^{t,y}_r)|X_r^{\xi}|^p \nonumber\\
	&+p\xi_r U^{t,y}_r \sgn(X_r^{\xi})|X_r^{\xi}|^{p-1})\big\}\,dr -\int_t^{s}Z^{t,y}_r|X_r^{\xi}|^p\,dW^{\vartheta}_r .
	\end{align}
	In what follows, we show that $(\xi^*,\vartheta^*)$ is a saddle point of the functional $\tilde J$, i.e.
	$$\tilde J(t,y,x;\xi^*,\vartheta)\leq\tilde J(t,y,x;\xi^*,\vartheta^*)\leq \tilde J(t,y,x;\xi,\vartheta^*).$$
	
	\textsc{Step 1: $\tilde J(t,y,x;\xi^*,\vartheta^*)\leq \tilde J(t,y,x;\xi,\vartheta^*)$ for every $\xi$. }

Set $\pi^*_s=\mathcal E(\int^s_t \vartheta^*_rdW_r).$  From the definition of $\vartheta^*$ in \eqref{optimal-control}, we see that $|\vartheta^*_r|\leq (1+\alpha)\theta^\alpha|Z^{t,y}_r|^{\alpha}.$ Using the estimate in \eqref{Z-estimate}, 
\begin{equation}\label{vartheta-estimate}
\begin{aligned}
\int^T_t |\vartheta^*_s|^2ds&\leq \int^T_{T-\delta} |\vartheta^*_s|^2ds+\int^{T-\delta}_t |\vartheta^*_s|^2ds\\
&\leq(1+\alpha)^2\theta^{2\alpha}\left(\int^T_{T-\delta} \frac{C^{\alpha}}{(T-s)^{2\alpha/\beta}}ds+\int^{T-\delta}_t C^{2\alpha}(1+\frac{1}{\delta^{1/\beta}})^{2\alpha}ds\right)\\
&\leq (1+\alpha)^2\theta^{2\alpha}\left(C^{\alpha}\delta^{1-2\alpha/\beta}+TC^{2\alpha}(1+\frac{1}{\delta^{1/\beta}})^{2\alpha}\right)<+\infty.
\end{aligned}
\end{equation}
Hence $\mathbb E[(\pi^*_s)^k]<+\infty$ for every $k>1.$ This allows us to show that the  stochastic integral in \eqref{ito} is a $Q^*$-martingale.  Since $Z^{t,y}\in H^q_{\mathcal{F}}(t,T^-;\mathbb R^{d})$ for $q\geq 2$ and
\begin{equation*}
\mathbb E[\sup_{t\leq r\leq s}|X_r^{\xi}|^{2p}]\leq C\left(|x|^{2p}+\mathbb E[\int^T_t|\xi_r |^{2p}\,dr\right),
\end{equation*}
we have that
	\begin{align*}
	\mathbb E_{Q^*}\left[\int_t^{s}|Z^{t,y}_r|^2|X_r^{\xi}|^{2p}\,dr\right]^{1/2}&=\mathbb E\left[(\pi^*_s)^{2}\int_t^{s}|Z^{t,y}_r|^2|X_r^{\xi}|^{2p}\,dr\right]^{1/2}\\	
	&\leq \left(\mathbb E\left[(\pi^*_s)^{2}\sup_{t\leq r\leq s}|X_r^{\xi}|^{2p}\right]^{3/4}\right)^{2/3}\left(\mathbb E\left[ \int_t^{s}|Z^{t,y}_r|^2\,dr\right]^{3/2}\right)^{1/3}\\
	&\leq \left(\mathbb E\left[\frac{(\pi^*_s)^{6}}{4}+\frac{3\sup_{t\leq r\leq s}|X_r^{\xi}|^{2p}}{4}\right]\right)^{2/3}\left(\mathbb E\left[ \int_t^{s}|Z^{t,y}_r|^2\,dr\right]^{3/2}\right)^{1/3}\\
	&<+\infty.
	\end{align*} 

Set $$c(y,x,\xi):=\eta(y)|\xi|^p+\lambda(y)|x|^p,\quad C(y,x,\xi,\vartheta):=c(y,x,\xi)-\frac{1}{\theta}h(\vartheta)|x|^p.$$ Thus,
	\begin{align}
	U^{t,y}_t|x|^p  
	&=\mathbb E_{Q^*}\left[U^{t,y}_s|X_{s}^{\xi}|^p\right] +\mathbb E_{Q^*}\left[\int_t^{s} C(Y_r^{t,y},X_r^{\xi},\xi_r,\vartheta^*_r)\,dr \right] \nonumber\\
	&\quad+\mathbb E_{Q^*}\left[\int_t^{s} \big\{ F(Y^{t,y}_{r},U^{t,y}_r)|X_r^{\xi}|^p+p\xi_r U^{t,y}_r \sgn(X_r^{\xi})|X_r^{\xi}|^{p-1}-c(Y_r^{t,y},X_r^{\xi},\xi_r)\big\}\,dr \right]\nonumber\\
	&\leq \mathbb E_{Q^*}\left[U^{t,y}_s|X_{s}^{\xi}|^p\right]  +\mathbb E_{Q^*}\left[\int_t^{s} C(Y_r^{t,y},X_r^{\xi},\xi_r,\vartheta^*_r))\,dr \right]. \label{suboptimal-estimate}	
	\end{align}

Since $U^{t,y}_t$ is nonnegative, we can obtain that
	\begin{equation*}
	\mathbb E_{Q^*}\left[\int_t^{s} \frac{1}{\theta}h(\vartheta^*_r)|X_r^{\xi}|^p\,dr \right]\leq \mathbb E_{Q^*}\left[U^{t,y}_s|X_{s}^{\xi}|^p\right]  +\mathbb E_{Q^*}\left[\int_t^{s} c(Y_r^{t,y},X_r^{\xi},\xi_r))\,dr \right].
	\end{equation*}
The right hand side is finite as $s$ goes to $T$ by the admissibility of $\xi$ and the uniform boundedness of $U^{t,y}$.  By Lemma~\ref{lemma-finite-costs},  letting $s \to T$ in \eqref{suboptimal-estimate} we get
	\begin{equation*} 
	v(t,y)|x|^p\leq \tilde J(t,y,x;\xi,\vartheta^{*}).
	\end{equation*}
	Finally note that the equality holds in \eqref{suboptimal-estimate} if $\xi=\xi^{*}$. 
	This yields
	\begin{equation*} \label{verification-in-expectation*}
	\begin{aligned}	
	v(t,y)|x|^p &=  \mathbb E_{Q^*}[v(s,Y_s^{t,y})|X_s^{\xi^*}|^p]+ \mathbb E_{Q^*}\left[\int_t^s C(Y_r^{t,y},X_r^{\xi^*},\xi^*_r,\vartheta^*_r)\,dr\right] \\
	& \rightarrow \tilde J(t,y,x;\xi^*,\vartheta^{*})   \quad \text{as } s \to T. 
	\end{aligned}
	\end{equation*}
	Thus,
	$$v(t,y)|x|^p=\tilde J(t,y,x;\xi^*,\vartheta^*)\leq \tilde J(t,y,x;\xi,\vartheta^*).$$ 
	\textsc{Step 2. $\tilde J(t,y,x;\xi^*,\vartheta)\leq\tilde J(t,y,x;\xi^*,\vartheta^*)$ for every $\vartheta$. }

	Let us introduce the sequence of stopping times
	$$\tau_n:=\inf\{r\in[t,T]:\int^r_t |\vartheta_{\rho}|^2 d\rho>n\}.$$
	Put $\vartheta^n_r=\vartheta_r \mathbb I_{r\leq \tau_n}$ and define
	$W^n_r=W_r+\int^r_t \vartheta^n_rdr.$ From the definition of $\tau_n,$ it follows that
	\begin{equation}
	\int^T_t |\vartheta^n_r|^2dr = \int^{\tau_n}_t |\vartheta_r|^2dr\leq n.
	\end{equation}
	Therefore, defining $\pi^n_s=\mathcal E(\int^s_t \vartheta^n_rdW_r),$ the Novikov condition implies that $E[\pi^n_T]=1.$ Setting $dQ^n=\pi^n_Td\mathbb P,$ by the Girsanov theorem $W^n$ is a Brownian motion under $Q.$ Moreover, $\mathbb E[(\pi^n_s)^k]<+\infty$ for every $k>1.$
	
	As discussed before, we can show that the  stochastic integrals $\int_t^{s}V^{t,y}_r|X_r^{\xi^*}|^p\,dW^{\vartheta^n}_r $ are $Q^n$-martingales for any $n\in \mathbb R$. Hence, we have 
		\begin{align}
	U^{t,y}_t|x|^p  
	&=\mathbb E_{Q^n}\left[U^{t,y}_s|X_{s}^{\xi^*}|^p\right] +\mathbb E_{Q^n}\left[\int_t^{s} C(Y_r^{t,y},X_r^{\xi^*},\xi^*_r,\vartheta^n_r)\,dr \right] \nonumber\\
	&\quad+\mathbb E_{Q^n}\left[\int_t^{s} \Big\{ \big(\theta^\alpha|Z^{t,y}_r |^{1+\alpha}-\vartheta^n_rZ^{t,y}_r+\frac{1}{\theta}h(\vartheta^n_r)\big)|X_r^{\xi^*}|^p\Big\}\,dr \right]\nonumber\\
	&\geq \mathbb E_{Q^n}\left[U^{t,y}_s|X_{s}^{\xi^*}|^p\right]  +\mathbb E_{Q^n}\left[\int_t^{s} C(Y_r^{t,y},X_r^{\xi^*},\xi^*_r,\vartheta^n_r))\,dr \right] .\label{suboptimal-estimate2}	
	\end{align}
	Letting $s \to T$ we get
\begin{equation*} 
U^{t,y}_t|x|^p   \geq  \mathbb E_{Q^n}\left[\int_t^{T} C(Y_r^{t,y},X_r^{\xi^*},\xi^*_r,\vartheta^n_r))\,dr \right] 
\end{equation*}
by Lemma \ref{lemma-finite-costs}. We can assume w.l.o.g. that  $\mathbb E_{Q}\left[\int_t^{T} \frac{1}{\theta}h(\vartheta_r)|X_r^{\xi^*}|^p\,dr \right]$ is finite. Since for $r\in[t,T], $
\begin{align*}
\mathbb E\left[\pi^n_rh(\vartheta^n_r)|X_r^{\xi^*}|^p\right]&\geq \mathbb E\left[\mathbb E\left[\pi^n_r h(\vartheta^{n-1}_r)|X_r^{\xi^*}|^p|\mathcal F_{r\wedge\tau_{n-1}}\right]\right]\\
&\geq \mathbb E\left[\mathbb E\left[\pi^n_r |\mathcal F_{r\wedge\tau_{n-1}}\right]h(\vartheta^{n-1}_r)|X_r^{\xi^*}|^p\right]\\
&=\mathbb E\left[\pi^{n-1}_r h(\vartheta^{n-1}_r)|X_r^{\xi^*}|^p\right].
\end{align*}
 the monotone convergence theorem yields,
$$\mathbb E_{Q^n}\left[\int_t^{T} \frac{1}{\theta}h(\vartheta^n_r)|X_r^{\xi^*}|^p\,dr \right]\rightarrow \mathbb E_{Q}\left[\int_t^{T} \frac{1}{\theta}h(\vartheta_r)|X_r^{\xi^*}|^p\,dr \right]$$
as $\tau_n\rightarrow \infty, Q\text{-}a.s..$ Hence,
	\begin{equation*} 
U^{t,y}_t|x|^p   \geq  \mathbb E_{Q}\left[\int_t^{T} C(Y_r^{t,y},X_r^{\xi^*},\xi^*_r,\vartheta_r))\,dr \right] 
\end{equation*}
	Choosing $n$ large enough, we have that $\vartheta^{n,*}=\vartheta^{*}$ recalling \eqref{vartheta-estimate}. Then equality holds in \eqref{suboptimal-estimate2} if $\vartheta=\vartheta^{*}$. Since $U^{t,y}$ is nonnegative, this implies that 
\begin{equation*} 
\mathbb E_{Q^*}\left[\int_t^{T} |\vartheta^*_r|^{1+1/\alpha}|X_r^{\xi^*}|^p\,dr \right]\leq  \mathbb E_{Q^*}\left[\int_t^{T} c(Y_r^{t,y},X_r^{\xi^*},\xi^*_r))\,dr \right]<+\infty.
\end{equation*}
Thus,
	$$\tilde J(t,y,x;\xi^*,\vartheta)\leq \tilde J(t,y,x;\xi^*,\vartheta^*)=v(t,y)|x|^p.$$ 
\end{proof}
\begin{remark}
It was shown that $(\xi^*,\vartheta^*)$ is a saddle point of the functional $\tilde J$, thus $(\xi^*,\vartheta^*)$ is indeed a solution of the robust control problem \eqref{value-function}. However, $\tilde J$ is not convex in $\xi$ for fixed $\vartheta$. So the saddle point $(\xi^*,\vartheta^*)$ may not be unique.
\end{remark}
\section{Asymptotic analysis}\label{asmyptotic analysis}

In this section, we give the proof of Theorem \ref{thm-asymptotic}.  The main idea is to construct a super- and subsolution to \eqref{pde-sup} by an asymptotic expansion around the benchmark solution and then to apply the comparison principle [Lemma \ref{lemma-comparison principle}]. 

The following lemma extends the results in \cite[Theorem 2.9]{Graewe2018}. The proof is given in the Appendix \ref{proof-lem-v_0}.
\begin{lemma}\label{lemma-v_0}
Let $\beta>2\alpha$. Under Assumptions (L.1)-(L.4), (F.1)-(F.3), the terminal value problem \eqref{v_0} admits a unique nonnegative solution $v_0$ in $C^{0,1}([0,T^-]\times\mathbb R^d)$. The solution satisfies the following estimates:
$$\frac{\underline{c}}{(T-t)^{1/\beta}}\leq v_0\leq \frac{C_0}{(T-t)^{1/\beta}},\quad  |D v_0|\leq \frac{C_0}{(T-t)^{1/\beta}},\quad (t,y)\in [0,T)\times\mathbb R^d,$$
for some constant $C_0>0.$
\end{lemma}

The next lemma establishes the existence of a unique nonnegative solution to the terminal value problem \eqref{w_1} and provides {\it a priori} estimates on the solution and its derivative. 
\begin{lemma}\label{lemma-w_1}
Let $\beta>2\alpha$. Under Assumptions (L.1)-(L.4), (F.1)-(F.3), the terminal value problem \eqref{w_1} admits a unique nonnegative viscosity solution $w_1$ in $C^{0,1}([0,T]\times\mathbb R^d)$.  Moreover,  the following estimates hold:
$$0\leq w_1\leq C_1(T-t)^{1-\alpha/\beta}, \quad |Dw_1|\leq C_1(T-t)^{1/2-\alpha/\beta}, \; (t,y)\in [0,T)\times\mathbb R^d,$$
for some constant $C_1>0.$
\end{lemma}
\begin{proof}
Set $A:=|\sigma Dv_0|^{1+\alpha}$ and $B:=\frac{(\beta+1) v_0^\beta}{\beta\eta ^\beta}.$
Let $\delta_0:=1/\|\frac{\mathcal L \eta}{\eta}\|.$ From \cite[Proposition 3.5]{Horst2018}, we know that for $(t,y)\in [T-\delta_0,T)\times\mathbb R^d,$  $$\frac{v_0(t,y)^{\beta}}{\eta(y)^\beta}\geq \frac{1-\|\frac{\mathcal L \eta}{\eta}\|(T-t)}{T-t}.$$
Hence, for $\delta:=\frac{\beta}{2(\beta+1)}\delta_0,$
\begin{equation}\label{delta-B}
B(t,y)=\frac{(\beta+1)v_0(t,y)^{\beta}}{\beta\eta(y)^\beta}\geq \frac{1+\beta/2}{\beta(T-t)}, \quad (t,y)\in [T-\delta,T)\times\mathbb R^d,
\end{equation}
and so
\begin{equation}\label{B-bound}
-B(t,y)+\frac{1}{\beta(T-t)}\leq \frac{1}{\beta\delta}\mathbb I_{t\in[0,T-\delta]}-\frac{1}{2(T-t)}\mathbb I_{t\in[T-\delta,T)}\leq \frac{1}{\beta\delta}  \quad (t,y)\in [0,T)\times\mathbb R^d.
\end{equation}
Using the estimates on $Dv_0$ in Lemma \ref{lemma-v_0} along with the fact that $\beta>2\alpha$, we have that
\begin{equation}\label{A-bound}
\mathbb E\left[\int^T_0 \left(A(s,Y_s^{t,y})(T-s)^{1/\beta}\right)^2\,ds\right]\leq \int^T_0 \frac{C}{(T-s)^{2\alpha/\beta}}\,ds<+\infty.
\end{equation}
By \eqref{B-bound} and \eqref{A-bound}, it follows from the Feyman-Kac formula \cite[Theorem 3.2]{Pardoux1999} that
\begin{equation*}
w_1(t,y):=\mathbb E\left[\int^T_t\exp\left(\int^s_t \left(-B(r,Y_r^{t,y}) +\frac{1}{\beta(T-r)}\right)\,dr\right)A(s,Y_s^{t,y})(T-s)^{1/\beta}\,ds\right]
\end{equation*}
is the unique viscosity solution to the terminal value problem \eqref{w_1} on $[0,T]\times\mathbb R^d$. Moreover, we have for $(t,y)\in [0,T)\times\mathbb R^d$ that
\begin{equation}\label{w_1-esti-1}
\begin{aligned}
w_1(t,y)\leq& \mathbb E\left[\int^T_t\exp\left(\int^s_t \frac{1}{\beta\delta}\,dr\right)A(s,Y_s^{t,y})(T-s)^{1/\beta}\,ds\right]\\
\leq& \int^T_t e^{T/(\beta\delta)}\frac{C}{(T-s)^{\alpha/\beta}}\,ds\\
\leq& C_1 (T-t)^{1-\alpha/\beta}
\end{aligned}
\end{equation}
for some constant $C_1$.

Next, we study the derivative of $w_1.$ For any $\varepsilon\in (0,T),$ restricting the PDE \eqref{w_1} to $[0,T-\varepsilon],$
\begin{equation*}\label{w_2-A1}
\left\{\begin{aligned}	&{-\partial_t v}(t,y)-\mathcal L v(t,y)-f_1(t,y,v(t,y))=0,    & (t,y)\in[0,T-\varepsilon)\times\mathbb R^d,&\\
& v(T-\varepsilon,y)=w_1(T-\varepsilon,y)  & y\in \mathbb R^d,&
\end{aligned}\right.
\end{equation*}
Since $A, B$ are bounded on $[0,T-\varepsilon],$ it follows from the Bismut-Elworthy formula  \cite[Theorem 4.2]{Fuhrman2002} that $w_1(t,\cdot)$ is differentiable for $t\in[0,T-\varepsilon]$ and
\begin{equation*}
\begin{aligned}
|Dw_1(t,y)|\leq &\frac{C}{(T-\varepsilon-t)^{1/2}}\|w_1(T-\varepsilon,\cdot)\|+\int^{T-\varepsilon}_t \frac{C}{(s-t)^{1/2}}\left((T-s)^{1/\beta}\|A(s,\cdot)\|\right.\\
&\left.+\left(\|B(s,\cdot)\|+\frac{1}{\beta(T-s)}\right) \|w_1(s,\cdot)\|\right)\, ds
\end{aligned}
\end{equation*}
for $(t,y)\in [0,T-\varepsilon)\times\mathbb R^d.$ Using the estimates on $v_0, w_1$ we get that
\begin{equation*}
\begin{aligned}
|Dw_1(t,y)|&\leq \frac{C}{(T-\varepsilon-t)^{1/2}}\varepsilon^{1-\alpha/\beta}+C\int^{T}_t \frac{1}{(s-t)^{1/2}}(T-s)^{-\alpha/\beta}\, ds\\
&\leq \frac{C}{(T-\varepsilon-t)^{1/2}}(T-t)^{1-\alpha/\beta},\quad (t,y)\in [0,T-\varepsilon)\times\mathbb R^d,
\end{aligned}
\end{equation*}
where $C$ is independent of $\varepsilon.$ By letting $\varepsilon$ go to zero, we see that (by an adjustment of $C_1$ if necessary)
\begin{equation}\label{Dv_2-estimate}
|Dw_1(t,y)|\leq C_1(T-t)^{1/2-\alpha/\beta},\quad (t,y)\in [0,T)\times\mathbb R^d.
\end{equation}
\end{proof}
By the transformation $v_1=\frac{1}{(T-t)^{1/\beta}}w_1$, we know that $v_1$ is a solution to the equation
 \begin{equation}\label{v_1}
-\partial_t v(t,y)-\mathcal L v(t,y)-|\sigma Dv_0|^{1+\alpha}+\frac{(\beta+1) v_0^\beta}{\beta\eta^\beta}v=0,    \quad (t,y)\in[0,T)\times\mathbb R.
\end{equation}
Moreover, since $\beta>2\alpha,$  there exists a constant $C_2>0$ such that for $(t,y)\in[0,T)\times\mathbb R^d,$
\begin{equation}\label{v_1-esti}
\begin{aligned}
0\leq v_1&\leq C_1(T-t)^{1-(1+\alpha)/\beta} \leq C_2(T-t)^{-1/\beta}, \\
|Dv_1|&\leq C_1(T-t)^{1/2-(1+\alpha)/\beta} \leq C_2(T-t)^{-1/\beta}.
\end{aligned}
\end{equation}
Armed with these estimates, we are now ready to prove the asymptotic result.
\begin{proof}[Proof of Theorem \ref{thm-asymptotic}]
Let $\delta$ be as in \eqref{delta-B} and set  $b:=\frac{{\bar C}^{\beta}}{(\beta+1)\underline{c}^{\beta}\delta^{1/\beta}}$.
Our goal is to find two constants $L_1>0, L_2<0$ such that $$u_i=v_0+\theta^{\alpha} v_1+\theta^{2\alpha} L_i\left(b +\frac{1}{(T-t)^{1/\beta}}\right), \quad i=1,2$$
 is a supersolution (i=1), respectively a subsolution (i=2) to \eqref{pde-sup}. For $i=1,2,$ 
\begin{align*}
&-\theta^\alpha|\sigma Du_i|^{1+\alpha}+\frac{ u_i^{\beta+1}}{\beta\eta ^\beta} -\lambda(y)\\
=&-\theta^\alpha|\sigma( Dv_0+\theta^\alpha Dv_1)|^{1+\alpha}+\frac{ \left(v_0+\theta^\alpha v_1+\theta^{2\alpha}L_i(b+\frac{1}{(T-t)^{1/\beta}})\right)^{\beta+1}}{\beta\eta ^\beta} -\lambda(y)\\
 =&-\theta^\alpha|\sigma Dv_0|^{1+\alpha}+\frac{ v_0^{\beta+1}+(\beta+1)\theta^\alpha v_0^{\beta} v_1}{\beta\eta ^\beta} -\lambda(y)+\theta^{2\alpha}\frac{L_i}{\beta(T-t)^{1/\beta+1}}+\mathcal I_i,
\end{align*}
where  $\mathcal I_i:=\mathcal{I}_i^0+\mathcal{I}_i^1+\mathcal{I}_i^2$ and $\mathcal{I}_i^0,\mathcal{I}_i^1,\mathcal{I}_i^2$ are given by
\begin{align*}
\mathcal{I}_i^0&:=-\theta^{2\alpha}L_i\frac{1}{\beta(T-t)^{1/\beta+1}};\\
\mathcal{I}_i^1&:=\frac{ \left(v_0+\theta^\alpha v_1+\theta^{2\alpha}L_i(b+\frac{1}{(T-t)^{1/\beta}})\right)^{\beta+1}}{\beta\eta ^\beta} -\frac{ v_0^{\beta+1}+(\beta+1)\theta^\alpha v_0^{\beta} v_1}{\beta\eta ^\beta}; \\
\mathcal{I}_i^2&:=\theta^\alpha|\sigma Dv_0|^{1+\alpha}-\theta^\alpha|\sigma(Dv_0+\theta^\alpha Dv_1)|^{1+\alpha}.
\end{align*}
It is sufficient to prove that $\mathcal I_1>0$ (supersolution) and that $\mathcal I_2<0$ (subsolution) on $[0,T)\times\mathbb R^d$.

The second order Taylor approximation around $v_0$ in the first summand of $\mathcal{I}_i^1$ yields a function $\zeta$ satisfying $\min\{v_0, u_i\}\leq\zeta\leq \max\{v_0, u_i\}$ such that
\begin{align*}
\mathcal{I}_i^1&=\theta^{2\alpha}L_i\frac{1}{\beta\eta ^\beta}(\beta+1)v_0^{\beta}(b+\frac{1}{(T-t)^{1/\beta}})+\frac{1}{2\eta ^\beta}(\beta+1)\zeta^{\beta-1}\left(\theta^\alpha v_1+\theta^{2\alpha}L_i(b+\frac{1}{(T-t)^{1/\beta}})\right)^2.
\end{align*}

The mean value theorem along with the triangle inequality also yields a constant $\tilde C_0>0$ such that
\begin{align*}
|\mathcal{I}_i^2|&\leq \theta^{2\alpha}\bar C^{\alpha}(|Dv_0|^{\alpha}+|Dv_0+\theta^\alpha Dv_1|^{\alpha})|Dv_1|\\
&\leq \theta^{2\alpha}\tilde{C_0}(T-t)^{(1+\alpha)/\beta}\\
&\leq \theta^{2\alpha}\frac{\tilde{C_0} T^{1-\alpha/\beta}}{(T-t)^{1/\beta+1}}.
\end{align*}
\textsc{Step 1: Construction of supersolution.} 
Using the lower bound of $v_0$ in Lemma \ref{lemma-v_0}, we have that
$$\frac{\eta(y)^\beta}{(\beta+1)v_0(t,y)^{\beta}(T-t)^{1/\beta+1}}\leq \frac{{\bar C}^{\beta}}{(\beta+1)\underline{c}^{\beta}(T-t)^{1/\beta}}\leq \frac{{\bar C}^{\beta}}{(\beta+1)\underline{c}^{\beta}\delta^{1/\beta}}=b.$$
Set $c:=\min\{\frac{1}{2},\frac{(\beta+1)\underline{c}^\beta}{\beta {\bar C}^\beta}\}.$ The preceding inequality along with the inequality \eqref{delta-B} yields 
\begin{equation}\label{L-K}
\begin{aligned}
-\frac{1}{\beta(T-t)^{1/\beta+1}}+\frac{1}{\beta\eta ^\beta}(\beta+1)v_0^{\beta}\left(b+\frac{1}{(T-t)^{1/\beta}}\right)&\geq c\frac{1}{(T-t)^{1/\beta+1}}.
\end{aligned}
\end{equation}
Since the second term in the definition of $\mathcal{I}_1^1$ is nonnegative,  we have that
\begin{equation*}
\begin{aligned}
\mathcal I_1&\geq c\theta^{2\alpha}\frac{L_1}{(T-t)^{1/\beta+1}}-\theta^{2\alpha}\frac{\tilde{C}_0 T^{1-\alpha/\beta}}{(T-t)^{1/\beta+1}}.
\end{aligned}
\end{equation*}
Choosing $L_1>\frac{\tilde{C_0}T^{1-\alpha/\beta}}{c},$ we obtain that
$\mathcal I_1>0.$ 
 
\textsc{Step 2: Construction of subsolution.} Using the lower bound of $v_0$ in Lemma \ref{lemma-v_0} again and choosing $L_2<0,\theta>0$ such that $\theta^{2\alpha}|L_2|(T^{1/\beta}b+1)\leq \frac{\underline{c}}{2}$, we obtain that $u_2\geq \frac{\underline{c}}{2(T-t)^{1/\beta}}\geq 0.$  Different from Step 1,  an additional estimate on the second term in the definition of $\mathcal{I}_2^1$ is needed to obtain that $\mathcal I_2<0.$
Since $\min\{v_0, u_2\}\leq\zeta\leq \max\{v_0, u_2\}$ , we see that $\zeta(T-t)^{1/\beta}$ can be bounded both from below and above.
Therefore, there exists a constant $\tilde{C}_1>0$ such that
$$\frac{1}{2\eta ^\beta}(\beta+1)\zeta^{\beta-1}\left(\theta^\alpha v_1+\theta^{2\alpha}L_i(b+\frac{1}{(T-t)^{1/\beta}})\right)^2\leq \theta^{2\alpha} \frac{\tilde{C}_1}{(T-t)^{1/\beta+1}}.$$
By the inequality \eqref{L-K} and the nonpositivity of $L_2$, we have that
\begin{equation*}
\begin{aligned}
-\frac{L_2}{\beta(T-t)^{1/\beta+1}}+\frac{1}{\beta\eta ^\beta}(\beta+1)v_0^{\beta}L_2(b+\frac{1}{(T-t)^{1/\beta}})&\leq c\frac{L_2}{(T-t)^{1/\beta+1}}.
\end{aligned}
\end{equation*}
Thus, 
\begin{equation*}
\begin{aligned}
\mathcal I_2&\leq c\theta^{2\alpha}\frac{L_2}{(T-t)^{1/\beta+1}}+\theta^{2\alpha}\frac{\tilde{C}_0 T^{1-\alpha/\beta}}{(T-t)^{1/\beta+1}}+\theta^{2\alpha} \frac{\tilde{C_1}}{(T-t)^{1/\beta+1}}<0
\end{aligned}
\end{equation*}
if we first choose $$L_2<-\frac{\tilde{C}_1+\tilde{C}_0T^{1-\alpha/\beta}}{c}$$
and then
$$\theta<\min\{1,\sqrt[2\alpha]{\underline{c}/(2|L_2|(T^{1/\beta}b+1))}\}.$$
Hence $u_2$ is a nonnegative viscosity subsolution to \eqref{pde-sup}. By Lemma \ref{lemma-comparison principle}, we then have that $u_2\leq v\leq u_1.$ Thus, the desired equality
\eqref{asymptotic-eq} follows from 
\begin{equation*}
\theta^\alpha w_1+\theta^{2\alpha}L_2( b(T-t)^{1/\beta}+1)\leq w-w_0\leq \theta^\alpha w_1+\theta^{2\alpha}L_1( b(T-t)^{1/\beta}+1).
\end{equation*}
\end{proof}

\appendix
\section{Appendix}
\subsection{Comparison principle}\label{Comparison principle}

In this section, we state and prove comparison principles for solutions to PDEs with superlinear gradient term. Both finite and singular terminal values will be considered.

We refer to \cite{Lio2010} as an important reference for PDEs with superlinear gradient term. The following comparison principle can be seen as a corollary to \cite[Theorem 3.1]{Lio2010}. 

\begin{proposition} \label{comparison-f}
Assume that (L.1)-(L.3) hold and $\phi\in C_m(\mathbb R^d).$	Let $v \in LSC_m([0,T]\times\mathbb R^d)$ and $u \in USC_m([0,T]\times\mathbb R^d)$ be a nonnegative viscosity super- and a nonnegative viscosity subsolution to the following PDE:
	\begin{equation} \label{pde-general}
	\left\{\begin{aligned}	&{-\partial_t v}(t,y)-\mathcal L v(t,y)-H(y,Dv(t,y))-f(t,y)=0,    & (t,y)\in[0,T)\times\mathbb R^d,&\\
	&v(T,y)=\phi(y)  & y \in\mathbb R^d,&
	\end{aligned}\right.
	\end{equation}	
	If $f\in C_m([0,T]\times\mathbb R^d)$, then $$u \leq v\quad \text{on} \quad [0,T]\times\mathbb R^d.$$
\end{proposition}

Let us now consider the more general PDE
\begin{equation} \label{pde-H}
\left\{\begin{aligned}	&{-\partial_t v}(t,y)-\mathcal L v(t,y)-H(y,Dv(t,y))- F(y,v(t,y))=0,    & (t,y)\in[0,T)\times\mathbb R^d,&\\
&v(T,y)=\phi(y),  & y \in\mathbb R^d.&
\end{aligned}\right.
\end{equation}
A comparison principle for such PDEs is obtained in \cite{Lio2010} under a Lipschitz continuity assumption of $F$ on $v$. This condition is not satisfied in our case; we only have monotonicity. Additional assumptions on the solution are thus required to establish a comparison principle. However, we can make a weaker assumption on the coefficients than (F.1) and (F.2).
\begin{itemize}
\item[(F.4)] The coefficients $\eta,\lambda,1/\eta:\mathbb R^d\rightarrow \mathbb [0,\infty)$ are continuous and $\lambda$ is of polynomial growth of order $m$.
\end{itemize}
We first introduce two subsets of functions having superlinear growth. For a given $r\geq 0,$ a function $h: I\times\mathbb R^d\rightarrow \mathbb R^d$ belongs to $\mathcal{SSG}^{\pm}_{r}$ if and only if
$$\liminf_{|y|\rightarrow\infty}\frac{\pm h(t,y)}{|y|^r}\geq 0.$$
Notice that $h\in \mathcal{SSG}^{+}_{r}$ (resp., $\mathcal{SSG}^{-}_{r}$) if, for any $\varepsilon>0,$ there exists $C_{\varepsilon}=C_{\varepsilon}(h)>0$ such that
$$h(t,y)\geq -\varepsilon|y|^r-C_{\varepsilon}(\text{resp., } h(t,y)\leq \varepsilon|y|^r+C_{\varepsilon}),  (t,y)\in I\times\mathbb R^d.$$
We define $\mathcal{SSG}_{r}=\mathcal{SSG}^{+}_{r}\cap\mathcal{SSG}^{-}_{r}.$ Notice that $h\in \mathcal{SSG}_{r}$ if and only if
$$\lim_{|y|\rightarrow\infty}\frac{| h(t,y)|}{|y|^r}=0$$ for every $t\in I.$
\begin{proposition} \label{comparison-general}
Assume that (L.1)-(L.3) and (F.4) hold and that $\phi\in C_m(\mathbb R^d).$ Let $v \in LSC([0,T]\times\mathbb R^d)\cap\mathcal{SSG}^{+}_{m}$ and $u \in USC([0,T]\times\mathbb R^d)\cap\mathcal{SSG}^{-}_{m}$ be a nonnegative viscosity super- and a nonnegative viscosity subsolution to \eqref{pde-H}. Suppose that there exists $\hat C>0$ such that for all $(t,y)\in[0,T]\times\mathbb R^d,$
	\begin{equation}\label{estimate_eta}
	u^{\beta+1}(t,y), v^{\beta+1}(t,y)\leq \hat C\eta^{\beta}(y)\langle y\rangle^m.
	\end{equation}
	Then,  $$u \leq v\quad \text{on} \quad [0,T]\times\mathbb R^d.$$
\end{proposition}
\begin{proof}
	\textsc{Step 1: linearization.} For $\rho\in (0,1),$ it is easy to verify that $\tilde v:=\rho v$ is a viscosity supersolution of the following PDE:
	\begin{equation*}
	\left\{\begin{aligned}
	&-\partial_t \tilde v(t,y)-\mathcal L \tilde v(t,y)-\rho H(y,\frac{D\tilde v(t,y)}{\rho})-\rho F(y,\frac{\tilde v(t,y)}{\rho})=0,     &(t,y)\in[0,T)\times\mathbb R^d,&\\
	&\tilde v(T,y)= \rho \phi(y),\quad &y \in\mathbb R^d.&
	\end{aligned}\right.
	\end{equation*}
	In what follows, we show that $w:=u-\tilde v$ is a viscosity subsolution of the following extremal PDE:
	\begin{equation}\label{w-eq}
	-\partial_t \upsilon(t,y)-\mathcal L \upsilon(t,y)-(\frac{1-\rho}{2})^{-\alpha}\bar C^{\alpha+1}| D\upsilon|^{\alpha+1}-(1-\rho)\left[\lambda(\bar y)+\frac{1+\beta}{\beta}{\hat C}\langle y\rangle^m\right]=0,
	\end{equation}
	for $ (t,y)\in[0,T)\times\mathbb R^d\cap \{w>0\}$. 
	
	
Let $\varphi\in C^2([0,T)\times\mathbb R^d)$ be a test function and $(\bar t,\bar y)\in [0,T)\times\mathbb R^d\cap \{w>0\}$ be a local maximum of $w-\varphi$. We may assume that this maximum is strict in the set $[\bar t-r,\bar t+r]\times\bar B_r(\bar y)\subset[0,T)\times\mathbb R^d$ for small $r\in (0,1)$; we choose $[0,r]\times\bar B_r(\bar y)$ if $\bar t=0$. Let
	$$\Phi(t,x,y):=\frac{|x-y|^2}{2\varepsilon}+\varphi(t,x)$$
	and 
	$$M_{\varepsilon}:=\max\limits_{t\in[\bar t-r,\bar t+r],x,y\in \bar B_r(\bar y)}\big(u(t,x)-\tilde v(t,y)-\Phi(t,x,y)\big).$$ This maximum is attained at a point $(t_{\varepsilon},x_{\varepsilon},y_{\varepsilon})$ and is strict. We know that 
	$$\frac{|x_{\varepsilon}-y_{\varepsilon}|^2}{2\varepsilon}\rightarrow 0 \mbox{ and } 
	M_{\varepsilon}\rightarrow u(\bar t,\bar y)-\tilde v(\bar t,\bar y)-\varphi(\bar t,\bar y)
	\mbox{ as } \varepsilon \rightarrow 0.$$

We now apply \cite[Theorem 8.3]{Crandall1992}. In terms of their notation we have that $k=2, u_1=u, u_2=-\tilde{v},\varphi(t,x,y)=\Phi(t,x,y)$. We also recall that $\bar{\mathcal P}^{2,-} (\tilde v) =-\bar{\mathcal P}^{2,+} (-\tilde v).$ Then, setting $p_{\varepsilon}=\frac{x_{\varepsilon}-y_{\varepsilon}}{_{\varepsilon}}$, we have that
	\begin{align*}
	\partial_x\Phi(t_{\varepsilon},x_{\varepsilon},y_{\varepsilon}) & =p_{\varepsilon}+D\varphi(t_{\varepsilon},x_{\varepsilon}),\\ 
	-\partial_y\Phi(t_{\varepsilon},x_{\varepsilon},y_{\varepsilon}) & =p_{\varepsilon}
	\end{align*}
	and that 
	\begin{align*}
	A =D^2\Phi(t_{\varepsilon},x_{\varepsilon},y_{\varepsilon})=\begin{pmatrix}
	& D^2\varphi(t_{\varepsilon},x_{\varepsilon})+\frac{I}{\varepsilon}&-\frac{I}{\varepsilon}\\
	&-\frac{I}{\varepsilon}&\frac{I}{\varepsilon}\\
	\end{pmatrix}.
	\end{align*}
From this we conclude that for every $\iota>0,$ there exist $a_1, a_2\in\mathbb R, X, Y\in \mathcal S_d$ such that $$(a_1,p_{\varepsilon}+D\varphi(t_{\varepsilon},x_{\varepsilon}),X)\in\bar{\mathcal P}^{2,+} u(t_{\varepsilon},x_{\varepsilon}), \quad(a_2,p_{\varepsilon},Y)\in\bar{\mathcal P}^{2,-}\tilde v(t_{\varepsilon},y_{\varepsilon}),$$
	such that $a_1-a_2=\partial_t\Phi(t_{\varepsilon},x_{\varepsilon},y_{\varepsilon})=\varphi_t(t_{\varepsilon},x_{\varepsilon})$ and such that 
	\begin{equation}\label{XY}
	-(\frac{1}{\iota}+\|A\|)I\leq\left(\begin{matrix}
	X&0\\
	0&-Y
	\end{matrix}\right)\leq A+\iota A^2.
	\end{equation}

	From the definition of viscosity solution, we obtain that 
	$$-a_1-b(x_{\varepsilon}) (p_{\varepsilon}+D\varphi(t_{\varepsilon},y_{\varepsilon}))-\frac{1}{2}\textit{tr}\left[\sigma\sigma^\ast(x_{\varepsilon})X\right]-H(x_{\varepsilon}, p_{\varepsilon}+D\varphi(t_{\varepsilon},x_{\varepsilon}))-F(x_{\varepsilon},u)\leq 0$$
	and that
	$$-a_2-b(y_{\varepsilon}) p_{\varepsilon}-\frac{1}{2}\textit{tr}\left[\sigma\sigma^\ast(y_{\varepsilon})Y\right]-\rho H(y_{\varepsilon}, \frac{p_{\varepsilon}}{\rho})- \rho F(y_{\varepsilon},\frac{\tilde v}{\rho})\geq 0.$$
	Substracting the two inequalities, we have
	\begin{equation*}
	\begin{split}
	&-\partial_t\varphi_t(t_{\varepsilon},y_{\varepsilon})+b(y_{\varepsilon}) p_{\varepsilon}-b(x_{\varepsilon}) (p_{\varepsilon}+D\varphi(t_{\varepsilon},y_{\varepsilon}))\\
	&+\frac{1}{2}\textit{tr}\left[\sigma\sigma^\ast(y_{\varepsilon})Y\right]-\frac{1}{2}\textit{tr}\left[\sigma\sigma^\ast(x_{\varepsilon})X\right]\\
	&+\rho F(y_{\varepsilon},\frac{\tilde v}{\rho})-F(x_{\varepsilon},u)
	\leq H(x_{\varepsilon}, p_{\varepsilon}+D\varphi(t_{\varepsilon},x_{\varepsilon}))-\rho H(y_{\varepsilon}, \frac{p_{\varepsilon}}{\rho}).
	\end{split}
	\end{equation*}
	We are now going to estimate the terms involving the drift, the volatility, and the functions $F$ and $H$ separately. 
	\begin{itemize}
	\item Since $b$ is Lipschitz continuous, 
	\begin{equation*}
	\begin{split}
	b(y_{\varepsilon}) p_{\varepsilon}-b(x_{\varepsilon}) (p_{\varepsilon}+D\varphi(t_{\varepsilon},y_{\varepsilon}))&=-b(x_{\varepsilon})D\varphi(t_{\varepsilon},y_{\varepsilon})+(b(y_{\varepsilon})-b(x_{\varepsilon}))p_{\varepsilon} \\
	&\geq -b(x_{\varepsilon})D\varphi(t_{\varepsilon},y_{\varepsilon})-\bar C\varepsilon^{-1}|x_{\varepsilon}-y_{\varepsilon}|^2.
	\end{split}
	\end{equation*}
	\item In order to estimate the volatility term we denote by $(e_i)_{1\leq i\leq d}$ the canonical basis of $\mathbb R	^d$. By using \eqref{XY} and the Lipschitz continuity of $\sigma$, we obtain
\begin{equation*}
	\begin{split}
	\textit{tr}\left[\sigma\sigma^\ast(x_{\varepsilon})X\right]-\textit{tr}\left[\sigma\sigma^\ast(y_{\varepsilon})Y\right]&	=\sum_{i=1}^{d}\langle X\sigma(x_{\varepsilon}),\sigma(x_{\varepsilon})\rangle-\sum_{i=1}^{d}\langle Y\sigma(y_{\varepsilon}),\sigma(y_{\varepsilon})\rangle\\
	&\leq \sum_{i=1}^{d}\langle D^2\varphi(t_{\varepsilon},x_{\varepsilon})\sigma(x_{\varepsilon}),\sigma(x_{\varepsilon})\rangle+\frac{1}{\varepsilon}|\sigma(x_{\varepsilon})-\sigma(y_{\varepsilon})|^2+\omega(\frac{\iota}{\varepsilon^2})\\
	&\leq \textit{tr}\left[\sigma\sigma^\ast(x_{\varepsilon})D^2\varphi(t_{\varepsilon},x_{\varepsilon})\right]+{\bar C}^2\varepsilon^{-1}|x_{\varepsilon}-y_{\varepsilon}|^2+\omega(\frac{\iota}{\varepsilon^2})
	\end{split}
	\end{equation*}
	where $\omega$ is a modulus of continuity which is independent of $\iota$ and $\varepsilon$. 
	\item We now estimate $\tilde F:=\rho F(y_{\varepsilon},\frac{\tilde v}{\rho})-F(x_{\varepsilon},u)$. To this end, we first observe that 
	$$u(t_{\varepsilon},x_{\varepsilon})-\tilde v(t_{\varepsilon},y_{\varepsilon})-\varphi(t_{\varepsilon},x_{\varepsilon})\geq M_{\varepsilon}\geq u(\bar t,\bar y)-\tilde v(\bar t,\bar y)-\varphi(\bar t,\bar y).$$
Since $(\bar t,\bar y)\in  \{w>0\}$ and $\varphi$ is continuous,  we can fix  $r$ small enough to obtain that
	$$u(t_{\varepsilon},x_{\varepsilon})-\tilde v(t_{\varepsilon},y_{\varepsilon})\geq 0.$$
Recalling the definition of $F$ in \eqref{nonlinearity}, the fact that $F(y,\cdot)$ is decreasing on $\mathbb R_+$ and the fact that $\rho(1-\rho^{\beta})<(1+\beta)(1-\rho) $ for $ 0<\rho<1,$ this yields
	\begin{equation}\label{tilde-F}
	\begin{split}
	\tilde F&=\rho F(y_{\varepsilon},\frac{\tilde v}{\rho})-F(y_{\varepsilon},u)+F(y_{\varepsilon},u)-F(x_{\varepsilon},u)\\
	&\geq (\rho-1)\lambda(y_{\varepsilon})+\frac{|u|^{\beta+1}}{\beta\eta(y_{\varepsilon})^\beta}-\rho^{-\beta}\frac{|\tilde v|^{\beta+1}}{\beta\eta(y_{\varepsilon})^\beta}\\
	&\quad-\omega_{R}(|x_{\varepsilon}-y_{\varepsilon}|)\\
	&=(\rho-1)\lambda(y_{\varepsilon})+\frac{|u|^{\beta+1}}{\beta\eta(y_{\varepsilon})^\beta}-\frac{|\tilde v|^{\beta+1}}{\beta\eta(y_{\varepsilon})^\beta}\\
	&\quad-\rho(1-\rho^{\beta})\frac{|v|^{\beta+1}}{\beta\eta(y_{\varepsilon})^\beta}-\omega_{R}(|x_{\varepsilon}-y_{\varepsilon}|)
	\\
	& \geq -(1-\rho)\lambda(y_{\varepsilon})-(1+\beta)(1-\rho)\frac{|v|^{\beta+1}}{\beta\eta(y_{\varepsilon})^\beta}-\omega_{R}(|x_{\varepsilon}-y_{\varepsilon}|)\\
	&\geq -(1-\rho)\left[\lambda(y_{\varepsilon})+\frac{1+\beta}{\beta}{\hat C}\langle y_{\varepsilon}\rangle^m\right]-\omega_{R}(|x_{\varepsilon}-y_{\varepsilon}|)
	\end{split}
	\end{equation}
	where $\omega_{R}$ denotes the modulus of continuity with $R:=|\bar y|+r$.  

	\item We finally estimate $\tilde H:=H(x_{\varepsilon}, p_{\varepsilon}+D\varphi(t_{\varepsilon},x_{\varepsilon}))-\rho H(y_{\varepsilon}, \frac{p_{\varepsilon}}{\rho}).$ By convexity, we have, for $z_1,z_2\in\mathbb R^d$, that $$|z_1|^{\alpha+1}-\rho|\frac{z_2}{\rho}|^{\alpha+1}\leq (1-\rho)|\frac{z_1-z_2}{1-\rho}|^{\alpha+1}.$$ Hence,
	\begin{equation*}
	\begin{split}
	&H(x_{\varepsilon}, p_{\varepsilon}+D\varphi(t_{\varepsilon},x_{\varepsilon}))-\rho H(y_{\varepsilon}, \frac{p_{\varepsilon}}{\rho})\\
	&\leq
(1-\rho)\left|\frac{\sigma(x_{\varepsilon})(p_{\varepsilon}+D\varphi(t_{\varepsilon},x_{\varepsilon}))-\sigma(y_{\varepsilon})p_{\varepsilon}}{1-\rho}\right|^{\alpha+1}\\
&\leq (\frac{1-\rho}{2})^{-\alpha}\bar C^{\alpha+1}\left(\left|\sigma(x_{\varepsilon})D\varphi(t_{\varepsilon},x_{\varepsilon})\right|^{\alpha+1}+(|x_{\varepsilon}-y_{\varepsilon}|\cdot|p_{\varepsilon}|)^{\alpha+1}\right)
	\end{split}
	\end{equation*}
	where (L.2), (L.3) are used in the last inequality.
\end{itemize}
%
%

Denoting a generic modulus of continuity independent of $\iota$ and $\varepsilon$ by $\omega$, we thus get 
	\begin{equation*}
	\begin{split}
	&-\partial_t \varphi(t_{\varepsilon},y_{\varepsilon})-\mathcal L \varphi(t_{\varepsilon}, y_{\varepsilon})-(\frac{1-\rho}{2})^{-\alpha}\bar C^{\alpha+1} |D\varphi(t_{\varepsilon},y_{\varepsilon})|^{\alpha+1}-(1-\rho)\left[\lambda(y_{\varepsilon})+\frac{1+\beta}{\beta}{\hat C}\langle y_{\varepsilon}\rangle^m\right]\\
	&\leq \omega(\varepsilon)+\omega(\frac{\iota}{\varepsilon^2}).
	\end{split}
	\end{equation*}
	Letting first $\iota$ go to 0 and then sending $\varepsilon$ to 0, we finally conclude the desired viscosoity subsolution property of $w$. 

	\textsc{Step 2: smooth strict supersolution}. We are now going to construct smooth strict supersolutions to to \eqref{w-eq} on $[T-\tau,T)$ for some small $\tau > 0$. To this end, let $$\psi(t,y):=(1-\rho)C\langle y\rangle^{m}e^{L(T-t)}$$ where $L, C>0$ will be chosen later. 
		Since $\lambda,\phi\in C_m(\mathbb R^d)$ and $u \in \mathcal{SSG}^{-}_{m}([0,T]\times\mathbb R^d)$, we choose a large enough constant $\bar C$ such that for $\zeta=\lambda,\phi$
	$$\zeta(y)\leq \bar C\langle y\rangle^m, \quad y\in\mathbb R^d,$$
	and such that
\begin{equation}\label{u-bound}
u(t,y)\leq \bar C\langle y\rangle^m, \quad (t,y)\in[0,T]\times\mathbb R^d .
\end{equation}
	Note that 
	$$D\langle y\rangle^{m}=m\langle y\rangle^{m-2} y,\quad D^2\langle y\rangle^{m}=m\langle y\rangle^{m-4}\left(\langle y\rangle^2I+(m-2)y\otimes y\right).$$
	Since $b,\sigma$ grow at most linearly,
	\begin{equation*}
	\begin{split}
	\mathcal L \psi(t, y)&\leq (1-\rho)Ce^{L(T-t)}\left[\bar C(1+|y|)|D\langle y\rangle^{m}|+\bar C^2(1+|y|)^2|D^2\langle y\rangle^{m}|\right]\\
	&\leq (1-\rho)Ce^{L(T-t)}\left[2m\bar C\langle y\rangle^{m}+2m(m-1)\bar C^2\langle y\rangle^{m}\right]\\
	&\leq [2m\bar C+2m(m-1)\bar C^2]\psi(t,y).
	\end{split}
	\end{equation*}
	Recalling that $(m-1)(\alpha+1)=m,$ we have
	\begin{equation*}
	\begin{split}
	&(\frac{1-\rho}{2})^{-\alpha}\bar C^{\alpha+1}|D\psi(t , y)|^{\alpha+1}\\
	&= (\frac{1-\rho}{2})^{-\alpha}\bar C^{\alpha+1}\cdot(1-\rho)^{\alpha+1}C^{\alpha+1}e^{(\alpha+1)L(T-t)}|D\langle y\rangle^{m}|^{\alpha+1}\\
	&\leq[2^{\alpha}m^{\alpha+1}\bar C^{\alpha+1} C^{\alpha}e^{\alpha L(T-t)}]\psi(t,y)
	\end{split}
	\end{equation*}
	By condition (F.4),
	\begin{equation*}
	\begin{split}
	(1-\rho)\left[\lambda (y)+\frac{1+\beta}{\beta}{\hat C}\langle y\rangle^m\right]&\leq (1-\rho)\frac{1+2\beta}{\beta}\bar C\langle y\rangle^{m}\leq\frac{1+2\beta}{\beta}\frac{\bar C}{C}\psi(t,y)
	\end{split}
	\end{equation*}
	Choosing $C>\max\{2m\bar C+2m(m-1)\bar C^2,2^{\alpha}m^{\alpha+1}\bar C^{\alpha+1},\frac{1+2\beta}{\beta}\bar C\},$ we have 
	\begin{equation*}
	\begin{split}
	&-\partial_t \psi(t, y)-\mathcal L \psi(t, y)-(\frac{1-\rho}{2})^{-\alpha}\bar C^{\alpha+1}|D\psi(t , y)|^{\alpha+1}-(1-\rho)\left[\lambda (y)+\frac{1+\beta}{\beta}{\hat C}\langle y\rangle^m\right]\\
	&>\psi(t,y)\left[L-C-1-C^{\alpha+1}e^{\alpha L(T-t)}\right].
	\end{split}
	\end{equation*}
	Then taking $L>C+1+C^{\alpha+1}e,$ we get
	\begin{equation*}
	\begin{aligned}
	-\partial_t \psi(t, y)-\mathcal L \psi(t, y)-(\frac{1-\rho}{2})^{-\alpha}\bar C^{\alpha+1}|D\psi(t , y)|^{\alpha+1}-(1-\rho)\left[\lambda (y)+\frac{1+\beta}{\beta}{\hat C}\langle y\rangle^m\right]>0
	\end{aligned}
\end{equation*}	 
	for all $y\in\mathbb R^d$ and $t\in [T-\tau,T)$, where $\tau=\frac{1}{\alpha L}.$

	\textsc{Step 3: conclusions.} Since $w\in USC([T-\tau,T]\times\mathbb R^d)\cap\mathcal{SSG}^{-}_{m}$, the function $w-\psi$ attains its maximum at some point $(t, y) \in [T-\tau,T]\times\mathbb R^d.$ We claim that $t=T.$ Indeed, suppose to the contrary that $\bar t<T.$ Then, since $w$ is a viscosity subsolution of \eqref{w-eq}, by taking $\psi$ as a test function,  
	$$-\partial_t \psi(t, y)-\mathcal L \psi(t, y)-(\frac{1-\rho}{2})^{-\alpha}\bar C^{\alpha+1}|D\psi(t , y)|^{\alpha+1}-(1-\rho)\left[\lambda (y)+\frac{1+\beta}{\beta}{\hat C}\langle y\rangle^m\right]\leq 0.$$
	This contradicts the fact that $\psi$ is a strict supersolution.
	Thus, for all $(t,y)\in[T-\tau,T]\times \mathbb R^d,$
	$$w(t,y)-\psi(t,y)\leq w(T,y)-\psi(T,y)\leq (1-\rho)\phi(y)-(1-\rho)C\langle y\rangle^m\leq 0$$
	where the last inequality follows from $C>\bar C.$ In particular, $w(t,y)\leq\psi(t,y).$ Letting $\rho\rightarrow 1,$ we get $u\leq v$ on $[T-\tau,T]\times\mathbb R^d.$
	
	The preceding argument can be iterated on time intervals of the same length $\tau$. Indeed, let us choose $C, L, \tau$ as in Step 2 and put $$\psi(t,y):=(1-\rho)C\langle y\rangle^{m}e^{L(T-\tau-t)}$$ on $[T-2\tau, T-\tau]$. It follows by \eqref{u-bound} and the previously established inequality $u\leq v$ on $[T-\tau,T]\times\mathbb R^d$ that for all $y\in\mathbb R^d,$ $$w(T-\tau,y)=u(T-\tau,y)- \tilde v(T-\tau,y)\leq (1-\rho)u(T-\tau,y)\leq (1-\rho)\bar C\langle y\rangle^m .$$ Following the same arguments as above, we obtain that for all $(t,y)\in [T-2\tau, T-\tau]\times \mathbb R^d,$
	$$w(t,y)-\psi(t,y)\leq w(T-\tau,y)-\psi(T-\tau,y)\leq (1-\rho)\bar C\langle y\rangle^m-(1-\rho)C\langle y\rangle^m\leq 0.$$
These arguments can be iterated to complete the proof.
\end{proof}

\begin{remark}\label{remark-tildeF}
It is worth noting that the constant $\hat C$ in \eqref{w-eq} is exactly derived from the upper bound of $v$ in \eqref{estimate_eta} when estimating $\tilde F$ in \eqref{tilde-F}. We show below that using the constant derived from the upper bound of $u$ instead is also feasible. To this end, we estimate $\tilde F$ in the following way:
\begin{equation}
	\begin{split}
	\tilde F&=\rho F(x_{\varepsilon},\frac{\tilde v}{\rho})-F(x_{\varepsilon},u)+\rho F(y_{\varepsilon},\frac{\tilde v}{\rho})-\rho F(x_{\varepsilon},\frac{\tilde v}{\rho})\\
	&\geq (\rho-1)\lambda(x_{\varepsilon})+\frac{|u|^{\beta+1}}{\beta\eta(x_{\varepsilon})^\beta}-\rho^{-\beta}\frac{|\tilde v|^{\beta+1}}{\beta\eta(x_{\varepsilon})^\beta}-\omega_{R}(|x_{\varepsilon}-y_{\varepsilon}|)\\
	&\geq (\rho-1)\lambda(x_{\varepsilon})+(1-\rho^{-\beta})\frac{|u|^{\beta+1}}{\beta\eta(x_{\varepsilon})^\beta}+\rho^{-\beta}(\frac{|u|^{\beta+1}}{\beta\eta(x_{\varepsilon})^\beta}-\frac{|\tilde v|^{\beta+1}}{\beta\eta(x_{\varepsilon})^\beta})\\
	&\quad-\omega_{R}(|x_{\varepsilon}-y_{\varepsilon}|)\\
	&\geq -(1-\rho)\lambda(x_{\varepsilon})-2^\beta(1-\rho){\hat C}\langle x_{\varepsilon}\rangle^m-\omega_{R}(|x_{\varepsilon}-y_{\varepsilon}|),
	\end{split}
	\end{equation}
In the last inequality we used the facts that $u^{\beta+1}(t,y)\leq \hat C\eta^{\beta}(y)\langle y\rangle^m$ on $ [0,T]\times\mathbb R^d$ and $\rho^{-\beta}-1\leq 2^\beta(1-\rho)$ for $\rho\in (\frac{1}{2},1).$  
\end{remark}

The next lemma establishes a comparison principle for continuous solutions to \eqref{pde-sup} when imposed with a singular terminal time. The proof uses the shifting argument given in \cite{Graewe2018}. 

\begin{lemma} \label{lemma-comparison principle}
Assume that (L.1)-(L.3), (F.1) and (F.2) hold. Let $n$ be as in condition (F.1). Let $\underline v, \overline v\in C_n([0,T^-]\times\mathbb R^d)$ be a nonnegative viscosity sub- and a nonnegative viscosity supersolution to~\eqref{pde-sup}, respectively, such that 
\begin{equation*}
	\lim_{t\rightarrow T}\overline v(t,y)=+\infty \quad \text{locally uniformly on $\mathbb R^d$.}
\end{equation*}
Then, 
\begin{equation*}
\underline v\leq\overline v \qquad \text{in} \quad [0,T)\times\mathbb R^d.
\end{equation*}
In particular, there exists at most one nonnegative viscosity solution in $C_n([0,T^-]\times\mathbb R^d)$ to~\eqref{pde-sup}.
\end{lemma}

\begin{proof}
	Due to the time-homogeneity of the PDE in~\eqref{pde-sup}, viscosity (super-/sub-)solutions stay viscosity (super-/sub-)solutions when shifted in time. For any $\delta>0$, we define the difference function $w:[0,T-\delta)\times\mathbb R^d\rightarrow \mathbb R$ by
\[
	w(t,y):= \underline v(t,y) -\rho\overline v(t+\delta,y).
\]
Under assumptions (F.1) and (F.2), we have that $\underline v, \overline v$ belong to $\mathcal{SSG}_{m}$ and satisfy the condition \eqref{estimate_eta} in Proposition \ref{comparison-general} on $[0,T)\times\mathbb R^d$. Hence, we can use the similar argument as in the proof of Proposition \ref{comparison-general} to obtain that $w$ is a viscosity subsolution of the following PDE:
	\begin{equation}\label{w-eq-inf}
	-\partial_t u(t,y)-\mathcal L u(t,y)-(\frac{1-\rho}{2})^{-\alpha}\bar C^{\alpha+1}|Du|^{\alpha+1}-(1-\rho)\left[\lambda(\bar y)+\frac{1+\beta}{\beta}{\hat C}\langle y\rangle^m\right]=0,
	\end{equation}
	for $ (t,y)\in[0,T-\delta)\times\mathbb R^d\cap \{w>0\}$ and $\lim\limits_{t\rightarrow T-\delta}w(t,y)\leq (1-\rho) \underline v(T-\delta,y)$ for $y \in\mathbb R^d.$ 
In fact, Remark \ref{remark-tildeF} shows that we can get around the difficulty of the singularity of $\overline v(\cdot+\delta,\cdot)$ at time $t=T-\delta$ in this step. 
	Following Steps 2 and 3 in the proof of Proposition \ref{comparison-general}, we have that $\underline v(t,y) \leq \overline v(t+\delta,y)$ on $[0,T-\delta]\times\mathbb R^d.$ Finally, by letting $\delta\rightarrow 0$ we conclude that $\underline v\leq\overline v$ on $[0,T)\times\mathbb R^d$ by continuity of $\overline v.$

\end{proof}

\subsection{Proof of Proposition \ref{comparison}}\label{proof-comparison}

Under assumptions (F.1), (F.2) and  \eqref{interval}, the functions $(t,y)\mapsto  (T-t)^{1/\beta}\underline u(t,y),  (T-t)^{1/\beta}\overline u(t,y)$ satisfy the condition \eqref{estimate_eta} in Proposition \ref{comparison-general}. Let us fix $\rho\in \left(\sqrt[\beta]{\frac{\frac{1}{4}\beta+1}{\frac{1}{2}\beta+1}},1 \right)$ and consider the difference $$w:=\underline{u}-\rho\overline{u}\in USC_n([T-\delta,T^-]\times\mathbb R^d)\subset\mathcal{SSG}_{m}([T-\delta,T^-]\times\mathbb R^d).$$ 
 The proof of the following lemma is similar to that of Proposition \ref{comparison-general}.

\begin{lemma}
	The function $w$ is a viscosity subsolution to
	\begin{equation}\label{diff-pde} 
	\begin{split}
	&-\partial_t w(t,y)-\mathcal L w(t,y)-(\frac{1-\rho}{2})^{-\alpha}\bar C^{\alpha+1}|Dw|^{\alpha+1}-l(t,y)w(t,y)\\
	&-(1-\rho)\left[\lambda(y)+\frac{1+\beta}{\beta}\frac{{\hat C}\langle y\rangle^m}{(T-t)^{1/\beta+1}}\right]=0,   
	\quad (t,y)\in [T-\delta,T)\times\mathbb R^d
	\end{split}
	\end{equation}
	where
	$$l(t,y):=\frac{F(y,\underline{u}(t,y))-F(y,\rho\overline{u}(t,y))}{\underline{u}(t,y)-\rho\overline{u}(t,y)}\mathbb I_{\underline{u}(t,y)\neq\rho\overline{u}(t,y)}.$$
\end{lemma}

The next lemma constructs a local smooth strict supersolution to \eqref{diff-pde}. 

\begin{lemma}\label{bar}
	There exists $L,C,\tau>0$ such that
	$$\chi(t,y):=(1-\rho)\frac{e^{L(T-t)}C\langle y\rangle^m}{(T-t)^{1/\beta}}$$
	satisfies
	\begin{equation}\label{diff-upp}
	\begin{aligned}
	&\mathcal J[\chi]:=-\partial_t \chi(t,y)-\mathcal L \chi(t,y)-(\frac{1-\rho}{2})^{-\alpha}\bar C^{\alpha+1}|D\chi(t,y)|^{\alpha+1}+\frac{1+\frac{1}{4}\beta}{\beta(T-t)}\chi(t,y)\\
	&-(1-\rho)\left[\lambda(y)+\frac{1+\beta}{\beta}\frac{{\hat C}\langle y\rangle^m}{(T-t)^{1/\beta+1}}\right]>0,\quad (t,y)\in[T-\tau,T)\times\mathbb R^d.
	\end{aligned}
	\end{equation}
\end{lemma}
\begin{proof}
	Set $\psi(t,y):=(1-\rho)e^{L(T-t)}C\langle y\rangle^m.$ Analogous to the proof of Proposition \ref{comparison-general}, we have
	\begin{equation*}
	\begin{split}
	\mathcal L \chi(t, y)&\leq [2m\bar C+2m(m-1)\bar C^2]\frac{\psi(t,y)}{(T-t)^{1/\beta}},\\
	(\frac{1-\rho}{2})^{-\alpha}\bar C^{\alpha+1}|D\chi(t , y)|^{\alpha+1}&\leq[2^{\alpha}m^{\alpha+1}\bar C^{\alpha+1} C^{\alpha}e^{\alpha L(T-t)}]\frac{\psi(t,y)}{(T-t)^{(1+\alpha)/\beta}},\\
	(1-\rho)\left[\lambda (y)+\frac{1+\beta}{\beta}\frac{{\hat C}\langle y\rangle^m}{(T-t)^{1/\beta+1}}\right]&\leq \frac{{\bar C}}{C}\psi(t,y)+\frac{1+\beta}{\beta}\frac{\bar C}{C}\frac{\psi(t,y)}{(T-t)^{1/\beta+1}}.
	\end{split}
	\end{equation*}
	Choosing $C>\max\{2m\bar C+2m(m-1)\bar C^2,2^{\alpha}m^{\alpha+1}\bar C^{\alpha+1},8\frac{1+\beta}{\beta}{\bar C}\},$ we obtain that
	\begin{equation*}
	\begin{split}
	\mathcal J[\chi]>&\frac{L\psi}{(T-t)^{1/\beta}}-\frac{\psi}{\beta (T-t)^{1/\beta+1}}-\frac{C\psi}{(T-t)^{1/\beta}}-C^{\alpha+1}e^{\alpha L(T-t)}\frac{\psi}{(T-t)^{(1+\alpha)/\beta}}\\
	&+\frac{1+\frac{1}{4}\beta}{\beta(T-t)^{1/\beta+1}}\psi-\psi-\frac{\psi}{8(T-t)^{1/\beta+1}}\\
	>&\psi \left[\frac{L-C-T^{1/\beta}}{(T-t)^{1/\beta}}+\frac{1-8C^{\alpha+1}e^{\alpha L(T-t)}(T-t)^{1-\alpha/\beta}}{8(T-t)^{1/\beta+1}}\right]
	\end{split}
	\end{equation*}
	
	Taking $L>C+T^{1/\beta},$ we get $\mathcal J[\chi]>0$
	for all $y\in\mathbb R^d$ and $t\in [T-\tau,T)$, where $\tau=\min\{\frac{1}{\alpha L},(8C^{\alpha+1}e^1)^{(\alpha-\beta)/\alpha}\}.$
\end{proof}

The following lemma is key to the proof of the comparison principle. 

\begin{lemma}\label{Phi}
Let $\tau$ be as in Lemma \ref{bar}. The function 
	$$\Phi(t,y):= w(t,y)-\chi(t,y)$$
	is either nonpositive or attains its supremum at some point $(\bar t,\bar y)$ in $[T-\tau,T)\times\mathbb R^d.$
\end{lemma}
\begin{proof}
	Suppose that the supremum of $\Phi$ on $[T-\tau,T)\times\mathbb R^d$ is positive and denote by $(t_k,y_k)$ a sequence in $[T-\tau,T)\times\mathbb R^d$ approaching the supremum point.  
For the choice of $C$ in Lemma \ref{bar}, $\eta(y)<C\langle y\rangle^m$ for all $y\in \mathbb R^d$. Thus, the representation
	$$\Phi(t,y)=\frac{\left[\frac{\underline u(t,y)(T-t)^{1/\beta}}{\langle y\rangle^n}-\frac{\rho\overline u(t,y)(T-t)^{1/\beta}}{\langle y\rangle^n}\right]\langle y\rangle^n-(1-\rho) e^{L(T-t)}C\langle y\rangle^m}{(T-t)^{1/\beta}},$$
along with Condition \eqref{asympotic} and the fact that $n<m$ yields
	$$\limsup\limits_{t\rightarrow T} \Phi(t,y)=-\infty,  \textrm{ uniformly on }\mathbb R^d.$$
	Hence $\lim\limits_{k} t_k<T.$ Furthermore, $\lim\limits_{k} |y_k| < \infty$ because $w\in \mathcal{SSG}^{-}_{m}$. 
	As a result, the supremum is attained at some point $(\bar t, \bar y)$ because $\Phi$ is upper semicontinuous. This proves the assertion. 
\end{proof}

We are now ready to prove the comparison principle. 

\begin{proof}[Proof of Proposition \ref{comparison}]
	\textsc{Step 1: comparison on $[T-\tau,T).$} Let $\tau$ be as in Lemma \ref{bar}. We claim that the function $\Phi$ introduced in Lemma \ref{Phi} is nonpositive. It then follows that $\underline u\leq \overline u$ in $[T-\tau,T)\times\mathbb R^d$ by letting $\rho\rightarrow 1$. In view of Lemma \ref{Phi}, we just need to consider the case where $\Phi$ attains its supremum at some point $(\bar t,\bar y)\in[T-\tau,T)\times\mathbb R^d.$ Since $\chi$ is smooth and $w$ is a viscosity subsolution to \eqref{diff-pde}, we have
	\begin{equation}\label{chi-viscosity}
	\begin{split}
	&-\partial_t \chi(\bar t, \bar y)-\mathcal L \chi(\bar t, \bar y)-(\frac{1-\rho}{2})^{-\alpha}\bar C^{\alpha+1}|D\chi|^{\alpha+1}-l(\bar t, \bar y)w(\bar t, \bar y)\\
	&-(1-\rho)\left[\lambda(\bar y)+\frac{1+\beta}{\beta}\frac{{\hat C}\langle y\rangle^m}{(T-t)^{1/\beta+1}}\right]\leq 0.
	\end{split}
	\end{equation}
	By the mean value theorem and in view of condition \eqref{interval}, 
	\begin{equation}\label{diff-bound}
	\begin{split}
	l(t,y)&=\frac{F(y,\underline{u}(t,y))-F(y,\rho\overline{u}(t,y))}{\underline{u}(t,y)-\rho\overline{u}(t,y)}\mathbb I_{\underline{u}(t,y)\neq\overline{u}(t,y)}\\
	&\leq \partial_u F(y,\rho\sqrt[\beta]{\frac{\frac{1}{2}\beta+1}{\beta+1}}\frac{\eta(y)}{(T-t)^{1/\beta}})\\
	&\leq-\frac{1+\frac{1}{4}\beta}{\beta(T-t)}.
	\end{split}
	\end{equation}
	Thus, comparing \eqref{diff-upp} with \eqref{chi-viscosity} yields
	\begin{equation}
	l(\bar t, \bar y)w(\bar t, \bar y)>-\frac{1+\frac{1}{4}\beta}{\beta(T-t)}\chi(\bar t, \bar y)\geq l(\bar t, \bar y)\chi(\bar t, \bar y).
	\end{equation}
	Since $l\leq 0,$  we can conclude that $\Phi(\bar t, \bar y)\leq 0,$ and so $\Phi\leq 0.$ 
	
	\textsc{Step 2: Comparison on $[T-\delta, T).$} If $\tau > \delta$, then the proof is finished. Else, we can proceed as follows.  From the condition \eqref{interval}, 
	$$\underline u(t,y), \overline u(t,y)\leq  \frac{\hat C}{\tau^{1/\beta}}\eta(y),\quad t\in[T-\delta, T-\tau].$$
	Since we have already shown that $\underline u(T-\tau,\cdot)\leq \overline u(T-\tau,\cdot),$ an application of our general comparison principle [Proposition \ref{comparison-general}] shows that $\underline u\leq \overline u$ on $[T-\delta,T)\times\mathbb R^d$. 
\end{proof}

\subsection{Proof of Lemma \ref{lemma-v_0}}\label{proof-lem-v_0}
The existence of a  classical solution $v_0$ to \eqref{v_0} along with the stated estimates on $v_0$ has been proved in \cite{Graewe2018}; the gradient was not given in \cite{Graewe2018}. In what follows we analyze the $C^{0,1}$ regularity of $v_0$ under weaker assumptions. As discussed in \cite{Graewe2018}, we can plug the asymptotic ansatz
\begin{equation} \label{ansatz-A1}
	\qquad\qquad\qquad v(T-t,y)= \frac{\eta(y)}{t^{1/\beta}}+\frac{u(t,y)}{t^{1+1/\beta}}, \quad u(t,y)=  O(t^2) \text{ uniformly in $y$ as $t\rightarrow0$}.
\end{equation}
into \eqref{v_0} and consider instead the PDE
\begin{equation} \label{pde-A1}
\left\{
\begin{aligned} 
	\partial_tu(t,y)&=\mathcal L u(t,y)+ f(t,y,u(t,y)), & t>0\,,y\in\mathbb R^d,&\\
	u(0,y)&=0,&y\in\mathbb R^d.&
\end{aligned}\right.
\end{equation}
where
$$f(t,y,u):=t\mathcal L \eta(y)+t^{p}\lambda(y)-\frac{\eta(y)}{\beta}\sum_{k=2}^\infty\dbinom{\beta +1}{k} \left(\frac{u}{t\eta(y)}\right)^k.$$

We now show that this PDE admits a mild solution in $C^{0,1}([0,\delta]\times\mathbb R^d)$. To this end we consider, similarly to Section \ref{mild solution}, the space
$$E:=\{u\in C^{0,1}_b([0,\delta]\times \mathbb R^d): \|u(t,\cdot)\|+\|t^{1/2}Du(t,\cdot)\| =O(t^{2}) \textit{ as } t\rightarrow 0\}$$
endowed with the weighted norm
$$\|u\|_E=\sup_{0<t\leq \delta,\,\, y\in\mathbb R^d}\|t^{-2}u(t,y)\|$$
and define the operator 
\begin{equation*}
\Gamma[u](t,y)=\int^t_0 P_{t-s}[f(s,\cdot,u(s,\cdot))](y)ds
\end{equation*}
Let $R>0$ and $\delta\in(0, \underline{c}/R].$ Using arguments given in \cite[Section 4]{Graewe2018}, we see that for every $u$ in the closed ball $\overline B_E(R):=\{u\in E:\|u\|_E\leq \underline{c}/\delta\}$, the function $f(\cdot,u(\cdot))$ belongs to $C_b([0,\delta]\times\mathbb R^d)$. In particular,  the map $\Gamma$ is well defined on $\overline B_E(R)$. Moreover, there exists a constant $L>0$ independent of $\delta$ such that 
$$|f(t,y,u(t,y))-f(t,y,v(t,y))|\leq L|u(t,y)-v(t,y)|, u,v\in \bar B_E(R), (t,y)\in[0,\delta]\times\mathbb R^d.$$
Now we are ready to carry out the fixed point argument.

Let  $B(a,b):=\int^1_0r^{a-1}(1-r)^{b-1}dr$ be the Beta function with $a, b>0$. We choose
\begin{equation*}
	R=2(1+MB_0)\left(\|\mathcal L\eta\|+\|\lambda\|\right),
\end{equation*}
and
\begin{equation*}
	 \delta=\min\{{\underline{c}/R},\big(2L(1+MB_1)\big),1\},
\end{equation*}
where $L > 0$ is the Lipschitz constant given by Lemma~\ref{lemma-locally-lip} and $B_0:=B(2,\frac{1}{2}), B_1:=B(3,\frac{1}{2})$.

 Let $u, v\in \overline B_\Sigma(R).$ For $(t,y)\in [0,\delta]\times\mathbb R^d,$
 \begin{equation*}
\begin{aligned}
|\Gamma[u](t,y)-\Gamma[v](t,y)|&=\left|\int^t_0 P_{t-s}[f(s,\cdot,u(s,\cdot))-f(s,\cdot,v(s,\cdot))](y)ds\right|\\
&\leq \int^t_0 \left\|f(s,\cdot,u(s,\cdot))-f(s,\cdot,v(s,\cdot))\right\|ds\\
&\leq \int^t_0 L\left\|u(s,\cdot)-v(s,\cdot)\right\|ds\\
&\leq \delta Lt^2\left\|u-v\right\|_Eds
\end{aligned}
\end{equation*}
Similarly,
 \begin{equation*}
\begin{aligned}
|D\Gamma[u](t,y)-D\Gamma[v](t,y)|&=\left|\int^t_0 DP_{t-s}[f(s,\cdot,u(s,\cdot))-f(s,\cdot,v(s,\cdot))](y)ds\right|\\
&\leq M\int^t_0\frac{1}{(t-s)^{1/2}}\left\|f(s,\cdot,u(s,\cdot))-f(s,\cdot,v(s,\cdot))\right\|ds\\ 
&\leq \int^t_0 ML\frac{1}{(t-s)^{1/2}}\left(s^2\|u-v\|_E\right)ds\\
&\leq \delta t^{3/2}MLB_1|u-v\|_E.
\end{aligned}
\end{equation*}
Hence 
\begin{equation*}
\|\Gamma[u]-\Gamma[v]\|_\Sigma\leq \frac{1}{2	}\|u-v\|_E.
\end{equation*}
To show that $\Gamma$ maps $\overline B_\Sigma(R)$ into itself, note that  $\delta\leq 1$ implies $s^k\leq 1$ for all $k>0$ and $s\in[0,\delta]$. Hence, for every $t\in[0,\delta]$ 
\begin{align*}
|\Gamma[0](t,y)|&=\left|\int^t_0 P_{t-s}[f(s,\cdot,0)](y)ds\right|\\
&\leq  \int^t_0 \|s\mathcal L \eta+s^{p}\lambda\|\,ds\\
&\leq t^2(\|\mathcal L \eta\|+\|\lambda\|)
\end{align*}
and
\begin{align*}
|D\Gamma[0](t,y)|&=\left|\int^t_0 DP_{t-s}[F_0(s,\cdot,0,0)](y)ds\right|\\
&\leq  \int^t_0 \frac{1}{(t-s)^{1/2}}M\|s\mathcal L \eta+s^{p}\lambda\|\,ds\\
&\leq t^{3/2}MB_0(\|\mathcal L \eta\|+\|\lambda\|)
\end{align*}
Thus,
\begin{align*}
\|\Gamma[u]\|_E&\leq \|\Gamma[u]-\Gamma[0]\|_E+\|\Gamma[0]\|_E\leq R
\end{align*}
The operator $\Gamma$ is therefore a contraction from $\overline B_E(R)$ to itself. Hence, it has a unique fixed point $u$ in $\overline B_E(R)$. We conclude that Equation \eqref{pde-A1} admits a mild solution in $C_b^{0,1}([0,\delta]\times\mathbb R^d).$

 In view of the ansatz \eqref{ansatz-A1}, $v_0$ is a solution to \eqref{v_0} in $C_b^{0,1}([T-\delta,T^-]\times\mathbb R^d)$ and there exists a constant $C>0$ such that for $(t,y)\in[T-\delta,T)\times\mathbb R^d,$
 \begin{equation*}
 |Dv_0|\leq \frac{C}{(T-t)^{1/\beta}}.
 \end{equation*}
To establish an {\it a priori} estimate of $Dv_0$ on $[0,T-\delta]\times\mathbb R^d,$ we introduce the corresponding FBSDE system
\begin{equation*}
\left\{
\begin{aligned}
dY^{t,y}_s&=b(Y^{t,y}_s)ds+\sigma(Y^{t,y}_s)dW_s, \quad s\in[t,T-\delta]\\
dU^{t,y}_s&=-F(Y^{t,y}_s,U^{t,y}_s)ds+Z^{t,y}_sdW_s,\quad s\in[t,T-\delta]\\
Y^{t,y}_t&=y, U^{t,y}_{T-\delta}=v(T-\delta,Y^{t,y}_{T-\delta}).
\end{aligned}\right.
\end{equation*}
By \cite[Theorem 4.1]{Karoui1997}, for $0\leq t\leq r\leq T-\delta $, the map $y\mapsto U^{t,y}_{t}=v_0(t,y)$  is differentiable  and $Z^{t,y}_r=\sigma^\ast(Y^{t,y}_r)Dv_0(r,Y^{t,y}_r).$ The boundeness of $Dv_0$ can be obtained by the classical BSDE estimates.
To conclude, for a constant $C_0>0,$
 \begin{equation}
 |Dv_0|\leq \frac{C_0}{(T-t)^{1/\beta}},\quad (t,y)\in[0,T)\times\mathbb R^d.
 \end{equation}

\bibliographystyle{siam}
{\small
\bibliography{robust}
}

\end{document}